\documentclass[journal,onecolumn]{IEEEtran}
\usepackage{booktabs}

\usepackage{amsmath}                
\usepackage{amsthm}                 
\usepackage{amssymb}
\usepackage{thmtools}
\usepackage{kpfonts}
\usepackage[geometry]{ifsym}
\usepackage{dsfont}
\usepackage{multirow}
\usepackage[mathscr]{euscript}
\usepackage{graphicx}
\usepackage{color}
\usepackage{paralist}
\usepackage{framed}
\usepackage{float}
\usepackage{longtable}
\usepackage{cite}
\usepackage{enumitem}

\usepackage[colorlinks=true, citecolor=blue, linkcolor=blue]{hyperref}       
\usepackage[font=itshape]{quoting}
\usepackage{lscape}
\usepackage{makecell}
\usepackage[section]{placeins}

\usepackage{listings}
\usepackage{comment}
\usepackage{framed} 
\usepackage{nomencl} 
\makenomenclature

\declaretheoremstyle[%
  headfont=\bfseries,%
  headpunct={:},%
  notefont=\normalfont\bfseries,%
  notebraces={--~}{},
    qed=$\blacksquare$,
]{definitionstyle}
\theoremstyle{definition}
\declaretheorem[style=definitionstyle,name=Definition]{defn}
\declaretheorem[style=definitionstyle,name=Theorem]{thm}

\theoremstyle{definition}

\theoremstyle{plain}
\theoremstyle{remark}

\begin{document}
%
\title{A Weighted Least Squares Error \\ Hetero-functional Graph State Estimator \\ of the American Multi-modal Energy System}
\author{Dakota J. Thompson, Amro M. Farid
\thanks{Dakota Thompson is a Senior Data Scientist at ISO New England.}
\thanks{Amro M. Farid is the Alexander Crombie Humphreys Chair Professor of Economics in Engineering at the Department of Systems Engineering at Stevens Institute of Technology where he is also the director of the Laboratory for Intelligent Integrated Networks of Engineering Systems.  He is also the Principal Systems Scientist for the National Energy Analysis Centre at CSIRO.}}

\date{January, 2024}
\maketitle

\begin{abstract}
As one of the most pressing challenges of the 21st century, global climate change demands a host of changes across at least four critical energy infrastructures: the electric grid, the natural gas system, the oil system, and the coal system.  In the context of the United States,  this paper refers to this system-of-systems as ``The American Multi-Modal Energy System (AMES)".  These combined changes necessitate an understanding of the AMES interdependencies, both structurally and behaviorally, to develop and enact effective policies. This work focuses on behavioral analysis methods to provide examples of how to analyze system behavior and the critical matter and energy flows through the system.  Building upon past works, two regions of the AMES are modeled, and their behavior is analyzed using Hetero-functional Graph Theory (HFGT).  More specifically, the work presents a weighted least square error state estimation model of the AMES.  State estimation has played a major role in the operation and development of the American Electric Power System.  This work extends the state estimation analysis beyond the single-operand electric grid environment into the heterogeneous environment of the AMES.  Employing a data-driven and model-based systems engineering approach in combination with HFGT, a Weighted Least Squares Error Hetero-functional Graph State Estimation (WLSEHFGSE) optimization program is developed to estimate the optimal  flows of mass and energy through the AMES.  This work is the first to integrate state estimation methods with HFGT.  Furthermore, it demonstrates how such a WLSEHFGSE recovers the mass and energy flows in a system-of-systems like the AMES with asset-level granularity.  
\end{abstract}




\section{Introduction} \label{Introduction}
\subsection{Motivation}
As one of the most pressing challenges of the 21st century, global climate change demands a host of changes across at least four critical energy infrastructures: the electric grid, the natural gas system, the oil system, and the coal system.  In the context of the United States,  this paper refers to this system-of-systems as ``The American Multi-Modal Energy System (AMES)".  The demands of global climate change require not just mitigation but also adaptation strategies to meet the global target of less than a $2^{\circ}$C rise by 2050 \cite{Elmqvist:2019:00, IEA:2016:00, Birol:2013:00, Commission:2011:00, Rogelj:2016:00, Obergassel:2016:00,williams:2015:00,williams:2015:01,state-of-california:2017:00, IEA:2017:00}.  Therefore, as policies are developed to drive the sustainable energy transition forward, they must be developed to enable sustainable and resilient architectures with optimal flows of mass and energy between all four major infrastructures.  In effect, the need for decarbonization must be harmonized with the need for economic development, national energy security, and equitable energy access\cite{Moniz:2022:00,IPCC:2023:00,wec:2016:00}.  These combined requirements necessitate an understanding of the AMES' interdependencies, both structurally and behaviorally, to develop and enact effective policies.

To develop this in-depth understanding of the AMES interdependencies, detailed data models of the AMES structural and behavioral properties are required.  Structural analyses provide valuable insights about the degree of structural resilience to disruptions in the energy system; be they due to natural causes or nefarious attacks. Behavioral analyses provide valuable insights about the most critical flows of energy through the system.  By performing these analyses on heterogeneous interdependent models of the AMES as a whole, rather than individual energy systems, greater insights can be gleaned.  For example, a structural resilience analysis can identify the impacts of one energy infrastructure's failure(s) on another, preventing potential cascading failures. A behavioral analysis of interdependent systems can result in more efficient operation of the systems together rather than individually.  This has been shown in multiple instances by pairing the electric grid with other systems such as the energy-water nexus, the electric grid with the natural gas system, or the electric grid with the oil system\cite{Farid:2016:EWN-J29, Lubega:2014:EWN-J11, Lubega:2014:EWN-J12, Thompson:2019:00, Munoz-Hernandez:2013:00, Lubega:2014:EWN-T09, Lubega:2013:EWN-C36, Lubega:2014:EWN-C35, Lubega:2014:EWN-C34, Farid:2013:EWN-C15, Albert:2004:00, Dong:2015:00, Jean-Baptiste:2003:00, Kriegler:2018:00, Robert-Lempert:2019:00, Rogers:2013:00, Lara:2020:00, Moreno:2017:00}. 

Recognizing the need for data-driven multi-energy system models, the (American) National Science Foundation (NSF) put forth a call for ``research to develop and make available simulated and synthetic data on interdependent critical infrastructures (ICIs), and thus to improve understanding and performance of these systems"\cite{Johnson:2017:00}.  The NSF project entitled ``American Multi-Modal Energy System Synthetic \& Simulated Data (AMES-3D)" seeks to fill this void with an open-source structural and behavioral model of the AMES.  Adhering to a strong theoretical foundation in Model-Based Systems Engineering (MBSE) \cite{Dori:2015:00, Friedenthal:2011:00}, the AMES-3D project aims to develop a physically informed and data-driven model of the AMES structure and behavior that captures the interdependencies between subsystems with asset-level granularity.  The developed model must support both structural and behavioral analyses, allowing the results to be compared and contrasted.  To that end, Hetero-Functional Graph Theory (HFGT) \cite{Thompson:2023:07, Schoonenberg:2019:ISC-BK04} is employed to transform MBSE models (expressed in SysML \ cite {Weilkiens:2007:00, Friedenthal:2011:00}) into mathematical models that directly support the quantitative analysis of system structure and behavior.  As previous works focus on the structural analysis of the AMES \cite{Thompson:2021:00, Thompson:2022:00}, this paper focuses its attention on behavioral analysis.

Data-driven behavioral analyses of energy systems have largely been predicated on state estimation. In electric power systems, after the 1965 New York City blackout\cite{Hull:2011:00}, there was a recognition that grid operators needed greater situational awareness of power flows and voltage profiles.  State estimation consequently emerged as a foundational tool upon which grid codes, operating procedures, and real-time electricity markets are built\cite{Abur:2004:00,wu:1990:00,ashok:2012:00,shivakumar:2008:00}.  Similarly, the natural gas and oil systems have adopted state estimators to facilitate the operation of their pipeline networks\cite{Jalving:201800, Behrooz:2015:00, Gong:2007:00}.  Despite these siloed advances in state estimation in individual energy sectors, an asset-level and holistic understanding of multi-energy system behavior remains elusive.  As the energy transition advances with a host of changes across all four critical infrastructures and becomes increasingly more interdependent, the sector-specific state estimators must be generalized, converged, and integrated to address the American Multi-modal Energy System as a whole.  Specifically, this work presents a weighted least square error, hetero-functional graph, state estimation model of the American Multi-modal Energy System. 

\subsection{Literature Review}
While there have been attempts to model multi-energy systems and their large-scale, heterogeneous energy flows, the field remains relatively nascent \cite{Albert:2004:00, Dong:2015:00, Jean-Baptiste:2003:00, Kriegler:2018:00, Robert-Lempert:2019:00, Rogers:2013:00, Lara:2020:00, Moreno:2017:00, Mejia-Giraldo:2012:00}.  Initially, models were developed in response to the 1973 oil crisis and have since sought to mitigate global climate change through decarbonization pathways that introduce new energy streams, such as hydrogen, synthetic fuels, biofuels, and other renewable energy sources.  Despite demonstrating many practical benefits, these works introduce their own weaknesses, including the lack of asset-level granularity, the difficulty of use, and the presence of geographical specificity that encourages one-off use cases that may not be applicable to other regions.  The overwhelming majority of energy system models focus on a single energy network \cite{masters:2013:00,mokhatab:2012:00,Lurie:2009:00,EIA:2014:11,Priyanka:2021:00}.   More recently, work has been published that analyzes only a pair of systems together, such as coupling the electric grid with one of the other fossil fuel systems in the AMES\cite{Albert:2004:00, Dong:2015:00, Jean-Baptiste:2003:00, Kriegler:2018:00, Robert-Lempert:2019:00, Rogers:2013:00, Lara:2020:00, Moreno:2017:00}.  These works, however, neither include all four critical energy infrastructures nor do they extend to the entire American geography.  As an exception, the EIA developed a comprehensive model called the National Energy Modeling System (NEMS), which it uses to produce the (American) Annual Energy Outlook\cite{EIA:2020:00}.  Despite serving this important function and being publicly available, this software tool remains opaque and difficult to use.  The EIA website itself recognizes:``[The] NEMS is only used by a few organizations outside of the EIA. Most people who have requested NEMS in the past have found out that it was too difficult or rigid to use \cite{EIA:2017:01}".  Consequently, holistic multi-energy system models of the AMES remain a present need for open-source, citizen-science to inform sustainable energy policies.

\begin{figure}[h!]
\begin{center}
\includegraphics[width=0.95\textwidth]{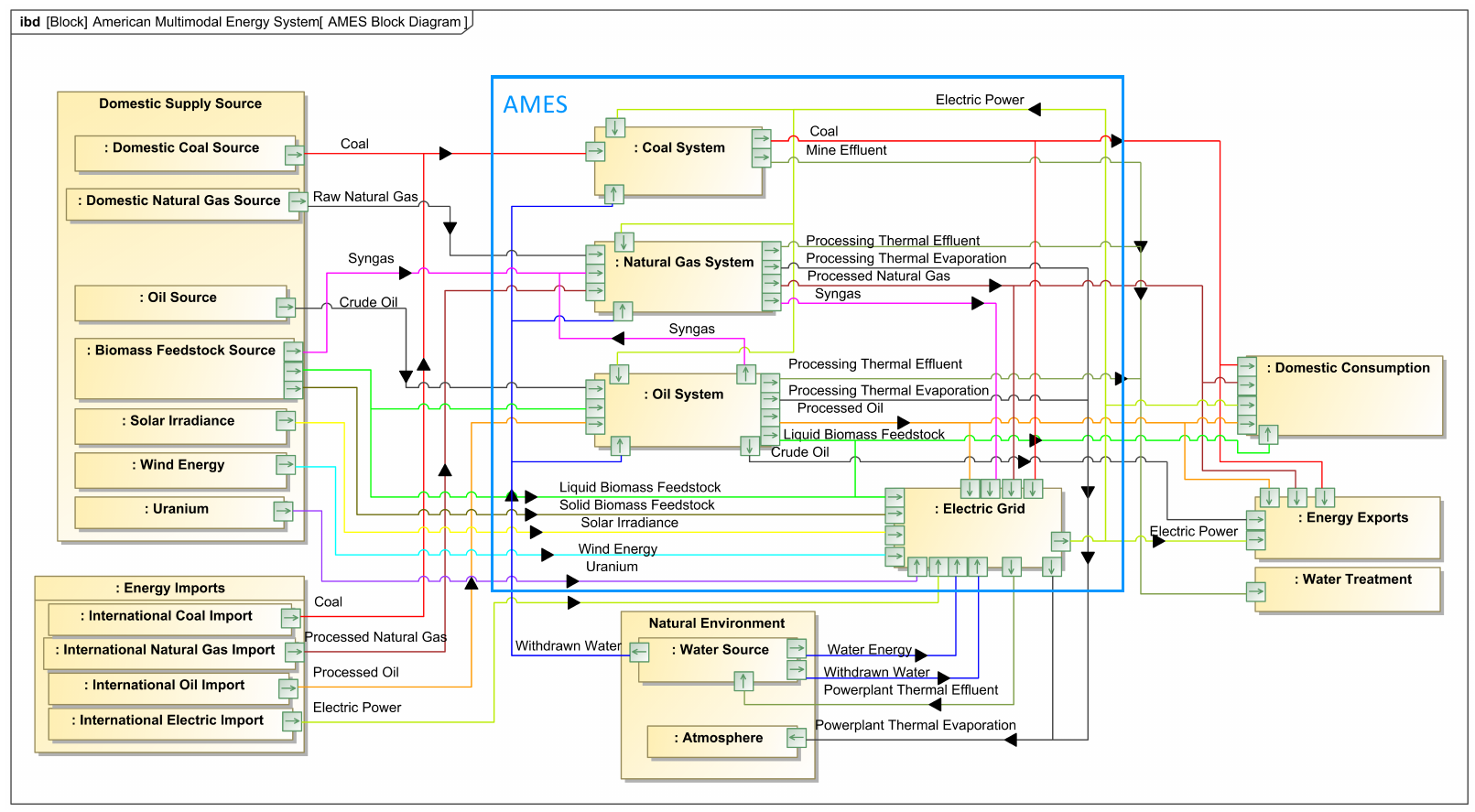}
\vspace{-0.1in}
\caption{The top-level context diagram of the AMES.  The domestic supply sources, energy imports, natural environment, domestic consumption, energy exports, and water treatment are external to the AMES' four subsystems of coal, natural gas, oil, and electric grid\cite{Thompson:2023:00}.}
\label{Fig:AMES-ContextDiagram}
\end{center}
\end{figure}

\begin{figure}[h!]\begin{center}\includegraphics[width=0.95\textwidth]{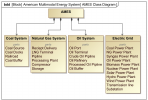}
\vspace{-0.1in}\caption{The top-level block definition diagram of the AMES.  The coal, oil, natural gas, and electric power systems are further decomposed into their constituent assets for further analysis\cite{Thompson:2023:00}.}\label{Fig:AMES-ClassDiagram}\end{center}\end{figure}

In order to manage the highly heterogeneous nature of the interactions within and between the AMES' component systems, the AMES-3D project has produced a reference architecture of the AMES \cite{Thompson:2023:00}.  It utilizes asset-level GIS data and openly available EIA datasets\cite {Thompson:2023:00} to identify and organize resources (i.e., facilities), processes, and interactions within and between the AMES' four subsystems.  The reference architecture is stated in the Systems Modeling Language (SysML) \cite{Friedenthal:2011:00} and explicitly depicts the flows of mass and energy.  Figure \ref{Fig:AMES-ContextDiagram} shows the top-level context diagram of AMES reference architecture, while Figure \ref{Fig:AMES-ClassDiagram} shows the top-level block-definition diagram\cite{Thompson:2023:00}.  These are then decomposed to elaborate the system form, function, and concept as the allocation of the latter onto the former\cite{Dori:2015:00, Friedenthal:2011:00}.  Unlike ``reference energy systems" that have been used in some national energy system optimization models\cite{PlazasNino:2022:00, Millot:2020:00, Dedinec:2016:00, Yangka:2016:00}, the development of a reference architecture for the AMES establishes this work on a solid MBSE foundation.  

To maintain self-consistency, the data sets used to infer the reference architecture (in SysML) are also used to produce the AMES' instantiated architecture as a mathematical model\cite{Thompson:2023:00, Thompson:2022:00}.  Hetero-functional graph theory (HFGT) provides a means by which to translate a SysML-based reference architecture to a mathematics-based instantiated architecture\cite{Thompson:2020:SPG-C68, Thompson:2021:00, Thompson:2022:00, Thompson:2022:ISC-C78}.  More specifically, it quantitatively interprets the SysML descriptions of system form, function, and concept into hetero-functional graphs.  Unlike traditional (i.e., formal) graphs, where noun-nodes (e.g., point-based energy facilities) are connected via noun-edges (e.g., line-based energy facilities), hetero-functional graphs (HFG) have interconnected nodes composed of subject-verb-object sentences called capabilities.  The serial and parallel connection of these capabilities reveals a ``story" where a wide variety of operands are transformed and transported as a flow through the HFG.  HFGs have been shown to provide more information than formal graphs in small-scale electric power distribution systems \cite{Thompson:2020:SPG-C68}, the entire American electric power system\cite{Thompson:2021:00}, and ultimately the entire AMES.  The last of these emphasizes how these interdependencies vary across the vast geography of the AMES\cite{Thompson:2023:ISC-J53,Thompson:2024:ISC-J55}.  The development of the AMES instantiated architecture from the reference architecture provides a solid foundation from which to conduct structural analyses \cite{Thompson:2023:00, Thompson:2022:00}.  While these past works have focused on the structural analysis of the AMES, the literature (until now) has yet to produce a commensurate behavioral analysis.  

\subsection{Original Contributions}
This paper uses a data-driven, MBSE-guided approach to develop the first open-source, asset-level granularity, behavioral model of the entire American Multi-Modal Energy System.  More specifically, socio-economic data on the flows of mass and energy through the AMES from the EIA are applied to the AMES reference architecture\cite{Thompson:2023:00} and instantiated architecture \cite{Thompson:2022:00} to create a behavioral model of the AMES.  The instantiated structural model includes the electric grid, the natural gas system, the oil system, the coal system, and the interconnections between them as defined by the AMES reference architecture, for the full contiguous United States of America (USA).  A Weighted Least Squares Error Hetero-functional Graph State Estimation (WLSEHFGSE) optimization is then developed to model the flows of mass and energy through AMES from openly available EIA data.  This WLSEHFGSE is presented as a formal generalization of the hetero-functional network minimum cost flow (HFNMCF) problem (introduced in Sec. \ref{Sec:HFNMCF}) and the weighted-least square error electric power system state estimator (WLSEEPSSE) problem (introduced in Sec.\ref{Sec:WLSEPSSE}). This WLSEHFGSE provides the first behavioral analysis of the entire AMES with an asset-level granularity on a monthly time scale.  It is fully extensible with respect to measured input datasets and it is capable of reconciling redundant datasets with conflicting information.  Ultimately, this behavioral model of the AMES allows direct comparisons in future studies between the structural \cite{Thompson:2022:00} and behavioral analyses.  

\subsection{Outline}

The remainder of the paper proceeds as follows.  Section \ref{BackgroundCh7} is a description of the background literature, which focuses on three areas: Hetero-Functional Graph Theory, 
Hetero-Functional Network Minimum Cost Flow Optimization (HFNMCF), and electric power systems state estimation.  Although this work makes every effort to be self-contained, background knowledge of MBSE\cite{Friedenthal:2011:00, Dori:2015:00}, HFGT\cite{Schoonenberg:2019:ISC-BK04,Farid:2022:ISC-J51}, and electric power system state estimation\cite{Monticelli:2012:00} is helpful.  Next, Sec. \ref{Sec:Methdology} presents the WLSEHFGSE as a formal generalization of the HFNMCF problem that admits measurement error of exogenous data.  Sec. \ref{Sec:WLSEHFGSE-AMES} then applies to the WLSEHFGSE to the AMES; taking into account its specific physical characteristics.   The section continues to identify all the data sets that are utilized to formulate the constraints included in the WLSEHFGSE optimization.  Section \ref{Sec:Results} then presents the case study results of the AMES' Western Interconnect and Eastern Regions.  Finally, the paper is brought to a conclusion in Section \ref{ConclusionsCh7}.

\section{Background}\label{BackgroundCh7}
The original contribution presented in this paper builds upon hetero-functional graph theory (HFGT), hetero-functional network minimum cost flow optimization (HFNMCF), and weighted-least square error electric power system state estimation (WLSEEPSSE). HFGT is based upon a highly descriptive set of ontological elements that describe a system's form, function, and concept as the allocation of the latter onto the former\cite{Dori:2015:00, Friedenthal:2011:00}\cite{Schoonenberg:2019:ISC-BK04}.  The heterogeneity of these elements facilitates direct application to complex systems-of-systems\cite{Schoonenberg:2019:ISC-BK04}.  Key definitions in HFGT are provided in Section \ref{Sec:HFGTDefinitions}.  HFNMCF optimizes the time-dependent flow and storage of multiple operands (or commodities) between buffers, allows for their transformation from one operand to another, and tracks the state of these operands.  It is presented in Section \ref{Sec:HFNMCF}.  Finally, WLSEEPSSE is presented as a single commodity state estimator that can be generalized to multi-energy systems.  It is presented in Section \ref{Sec:WLSEPSSE}.  
\subsection{Hetero-functional Graph Theory Key Definitions}\label{Sec:HFGTDefinitions}
This section provides key definitions from hetero-functional graph theory to support the methodological development of a WLSEHFGSE (in Sec. \ref{Sec:Methdology}).  For a more in-depth description of HFGT, readers are directed to past works\cite{Schoonenberg:2019:ISC-BK04, Thompson:2020:SPG-C68, Thompson:2021:00}, including the definition of the AMES reference architecture \cite{Thompson:2023:00}, and the development of the AMES' instantiated architecture using HFGT \cite{Thompson:2022:00}.  

\begin{defn}[System Operand \cite{SE-Handbook-Working-Group:2015:00}]\label{Defn:SO}
An asset or object $l_i \in L$ that is operated on or consumed during the execution of a process.  
\end{defn}

\begin{defn}[System Process\cite{Hoyle:1998:00,SE-Handbook-Working-Group:2015:00}]\label{def:CH7:process}
An activity $p \in P$ that transforms or transports a predefined set of input operands into a predefined set of outputs. 
\end{defn}

\begin{defn}[System Resource \cite{SE-Handbook-Working-Group:2015:00}]\label{Defn:SR}
An asset or object $r_v \in R$ that facilitates the execution of a process.  
\end{defn}

\begin{defn}[Buffer\cite{Schoonenberg:2019:ISC-BK04}]\label{defn:BSCh7}
    A resource $r_v \in R$ is a buffer $b_s \in B_S$ iff it is capable of storing or transforming one or more operands at a unique location in space.  
\end{defn}
\noindent Additionally, these resources are capable of one or more system processes including transform operands, hold operands, and transport operands \cite{Schoonenberg:2019:ISC-BK04}. 


\begin{defn}[Capability\cite{Farid:2006:IEM-C02,Farid:2007:IEM-TP00,Farid:2008:IEM-J05,Farid:2008:IEM-J04,Farid:2015:ISC-J19,Farid:2016:ISC-BC06}]\label{defn:capabilityCh7}
An action $e_{wv} \in {\cal E}_S$ (in the SysML sense) defined by a system process $p_w \in P$ being executed by a resource $r_v \in R$.  It constitutes a subject + verb + operand sentence of the form: ``Resource $r_v$ does process $p_w$".  
\end{defn}
\noindent In the context of the AMES, the operands are flows of matter and energy like coal, oil, natural gas, and electric power.  The resources include energy conversion and transportation facilities, such as electric power plants, refineries, electric power lines, and natural gas pipelines.  Of these, the buffers are the point facilities, such as electric power plants and refineries.  The capabilities are expressed in ``subject + verb + object" sentences, such as ``NG refinery refines raw natural gas."

The flow of operands through buffers and capabilities is described by the positive and negative hetero-functional incidence tensors (HFIT).

\begin{defn}[The Negative 3$^{rd}$ Order Hetero-functional Incidence Tensor (HFIT) $\widetilde{\cal M}_\rho^-$]\cite{Farid:2022:ISC-J49,Thompson:2022:ISC-C80} \label{Defn:N3OHFIT}
The negative hetero-functional incidence tensor $\widetilde{\cal M_\rho}^- \in \{0,1\}^{|L|\times |B_S| \times |{\cal E}_S|}$  is a third-order tensor whose element $\widetilde{\cal M}_\rho^{-}(i,y,\psi)=1$ when the system capability ${\epsilon}_\psi \in {\cal E}_S$ pulls operand $l_i \in L$ from buffer $b_{s_y} \in B_S$.
\end{defn} 

\begin{defn}[The Positive  3$^{rd}$ Order Hetero-functional Incidence Tensor (HFIT)$\widetilde{\cal M}_\rho^+$]\cite{Farid:2022:ISC-J49,Thompson:2022:ISC-C80} \label{Defn:P3OHFIT}
The positive hetero-functional incidence tensor $\widetilde{\cal M}_\rho^+ \in \{0,1\}^{|L|\times |B_S| \times |{\cal E}_S|}$  is a third-order tensor whose element $\widetilde{\cal M}_\rho^{+}(i,y,\psi)=1$ when the system capability ${\epsilon}_\psi \in {\cal E}_S$ injects operand $l_i \in L$ into buffer $b_{s_y} \in B_S$.
\end{defn}
\noindent These incidence tensors are straightforwardly ``matricized" to form 2$^{nd}$ Order Hetero-functional Incidence Matrices $M = M^+ - M^-$ with dimensions $|L||B_S|\times |{\cal E}|$. Consequently, the supply, demand, transportation, storage, transformation, assembly, and disassembly of multiple operands in distinct locations over time can be described by an Engineering System Net and its associated State Transition Function\cite{Farid:2022:ISC-J49, Schoonenberg:2017:IEM-J34}.

\begin{defn}[Engineering System Net]\label{Defn:ESN}
An elementary Petri net ${\cal N} = \{S, {\cal E}_S, \textbf{M}, W, Q\}$, where
\begin{itemize}
    \item $S$ is the set of places with size: $|L||B_S|$,
    \item ${\cal E}_S$ is the set of transitions with size: $|{\cal E}|$,
    \item $\textbf{M}$ is the set of arcs, with the associated incidence matrices: $M = M^+ - M^-$,
    \item $W$ is the set of weights on the arcs, as captured in the incidence matrices,
    \item $Q=[Q_B; Q_E]$ is the marking vector for both the set of places and the set of transitions. 
\end{itemize}
\end{defn}

\begin{defn}[Engineering System Net State Transition Function\cite{schoonenberg:2021:00}]\label{Defn:ESN-STF}
The  state transition function of the engineering system net $\Phi()$ is:
\begin{equation}\label{CH6:eq:PhiCPN}
Q[k+1]=\Phi(Q[k],U^-[k], U^+[k]) \quad \forall k \in \{1, \dots, K\}
\end{equation}
where $k$ is the discrete time index, $K$ is the simulation horizon, $Q=[Q_{B}; Q_{\cal E}]$, $Q_B$ has size $|L||B_S| \times 1$, $Q_{\cal E}$ has size $|{\cal E}_S|\times 1$, the input firing vector $U^-[k]$ has size $|{\cal E}_S|\times 1$, and the output firing vector $U^+[k]$ has size $|{\cal E}_S|\times 1$.  
\begin{align}\label{CH6:CH6:eq:Q_B:HFNMCFprogram}
Q_{B}[k+1]&=Q_{B}[k]+{M}^+U^+[k]\Delta T-{M}^-U^-[k]\Delta T \\ \label{CH6:CH6:eq:Q_E:HFNMCFprogram}
Q_{\cal E}[k+1]&=Q_{\cal E}[k]-U^+[k]\Delta T +U^-[k]\Delta T
\end{align}
where $\Delta T$ is the duration of the simulation time step.  This state transition function is incorporated directly into the HFNMCF problem in Sec. \ref{Sec:HFNMCF}.  
\end{defn}

Additionally, each operand can have its own state and evolution.  This behavior is described by an Operand Net and its associated State Transition Function for each operand.  
\begin{defn}[Operand Net\cite{Farid:2014:ISC-C37,Farid:2014:ISC-C38,Farid:2015:ISC-J19,Khayal:2017:ISC-J35,Schoonenberg:2017:IEM-J34}]\label{Defn:OperandNet} Given operand $l_i$, an elementary Petri net ${\cal N}_{l_i}= \{S_{l_i}, {\cal E}_{l_i}, \textbf{M}_{l_i}, W_{l_i}, Q_{l_i}\}$ where 
\begin{itemize}
\item $S_{l_i}$ is the set of places describing the operand's state.  
\item ${\cal E}_{l_i}$ is the set of transitions describing the evolution of the operand's state.
\item $\textbf{M}_{l_i} \subseteq (S_{l_i} \times {\cal E}_{l_i}) \cup ({\cal E}_{l_i} \times S_{l_i})$ is the set of arcs, with the associated incidence matrices: $M_{l_i} = M^+_{l_i} - M^-_{l_i} \quad \forall l_i \in L$.  
\item $W_{l_i} : \textbf{M}_{l_i}$ is the set of weights on the arcs, as captured in the incidence matrices $M^+_{l_i},M^-_{l_i} \quad \forall l_i \in L$.  
\item $Q_{l_i}= [Q_{Sl_i}; Q_{{\cal E}l_i}]$ is the marking vector for both the set of places and the set of transitions. 
\end{itemize}
\end{defn}

\begin{defn}[Operand Net State Transition Function]\label{Defn:OperandNet-STF}
The  state transition function of each operand net $\Phi_{l_i}()$ is:
\begin{equation}\label{CH6:eq:PhiSPN}
Q_{l_i}[k+1]=\Phi_{l_i}(Q_{l_i}[k],U_{l_i}^-[k], U_{l_i}^+[k]) \quad \forall k \in \{1, \dots, K\} \quad i \in \{1, \dots, L\}
\end{equation}
where $Q_{l_i}=[Q_{Sl_i}; Q_{{\cal E} l_i}]$, $Q_{Sl_i}$ has size $|S_{l_i}| \times 1$, $Q_{{\cal E} l_i}$ has size $|{\cal E}_{l_i}| \times 1$, the input firing vector $U_{l_i}^-[k]$ has size $|{\cal E}_{l_i}|\times 1$, and the output firing vector $U^+[k]$ has size $|{\cal E}_{l_i}|\times 1$.  

\begin{align}\label{X}
Q_{Sl_i}[k+1]&=Q_{Sl_i}[k]+{M_{l_i}}^+U_{l_i}^+[k]\Delta T - {M_{l_i}}^-U_{l_i}^-[k]\Delta T \\ \label{CH6:CH eq:Q_E:HFNMCFprogram}
Q_{{\cal E} l_i}[k+1]&=Q_{{\cal E} l_i}[k]-U_{l_i}^+[k]\Delta T +U_{l_i}^-[k]\Delta T
\end{align}
\end{defn}
\noindent This state transition function is incorporated directly into the MHFNCF problem in Sec. \ref{Sec:HFNMCF}.  

\subsection{The Hetero-Functional Network Minimum Cost Flow Optimization Problem (HFNMCF)}\label{Sec:HFNMCF} 
This section serves to introduce the reader to the HFNMCF problem advanced in 2021 and demonstrated on an integrated hydrogen-natural gas system\cite{schoonenberg:2021:00, Schoonenberg:2017:IEM-J34}.  As mentioned previously, the strategy for deriving the WLSEHFGSE (in Sec. \ref{Sec:Methdology}) is to recognize that it is a formal generalization of the HFNMCF problem. It optimizes the time-dependent flow and storage of multiple operands (or commodities) between buffers, allows for their transformation from one operand to another, and tracks the state of these operands.  In this regard, it is a very flexible optimization problem that applies to a wide variety of complex engineering systems.  For the purposes of this paper, the HFNMCF is a type of discrete-time-dependent, time-invariant, convex optimization program.

\vspace{0.2in}
\begin{align}\label{Eq:ObjFunc1}
\text{minimize } Z &= \sum_{k=1}^{K-1} f_k(x[k],y[k]) \\ \label{Eq:EqualityConstraints}
\text{s.t. } A_{CP}X &= B_{CP} \\ \label{ch6:eq:QPcanonicalform:3}
D_{CP}X &\leq E_{CP} \\ \label{Eq:DeviceModels}
g(X,Y)&=0 \\
h(Y) &\leq 0
\end{align}
where 
\begin{itemize}
\item $Z$ is a convex objective function separable in $k$.
\item $k$ is the discrete time index. 
\item $K$ is the simulation horizon.
\item $f()$ is a set of discrete-time-dependent convex functions.
\item $X=\left[x[1]; \ldots; x[K]\right]$  is the vector of primary decision variables at time $k$.
\begin{equation}
x[k] = \begin{bmatrix} Q_B ; Q_{\cal E} ; Q_{SL} ; Q_{{\cal E}L} ; U^- ; U^+ ; U^-_L ; U^+_L \end{bmatrix}[k] \quad \forall k \in \{1, \dots, K\}
\end{equation}
As the next section elaborates, the domain of $x[k]$ depends on the specific problem application and may include positive and negative reals and/or positive integers.
\item $Y=\left[y[1]; \ldots; y[K]\right]$  is the vector of auxiliary decision variables at time $k$.  The need for auxiliary decision variables depends on the presence and nature of the device models $g(X,Y)=0$ in Equation \ref{Eq:DeviceModels}.
\item $A_{CP}$ is the linear equality constraint coefficient matrix.
\item $B_{CP}$ is the linear equality constraint intercept vector.
\item $D_{CP}$ is the linear inequality constraint coefficient matrix.
\item $E_{CP}$ is the linear inequality constraint intercept vector.  
\item g(X,Y) and h(Y) are a set of (optional) device model functions whose presence and nature depend on the specific problem application.  As they are not further required in this paper, the interested reader is referred to \cite{schoonenberg:2021:00} for a deeper explanation.   
\end{itemize}
Equations \ref{Eq:ObjFunc1}-\ref{ch6:eq:QPcanonicalform:3} are elaborated in their respective subsections below.  

\vspace{0.1in}
\subsubsection{Objective Function}
With respect to the objective function in Eq.  \ref{Eq:ObjFunc1}, $Z$ is a convex objective function separable in $k$.  For the remainder of this work, the discrete-time-dependent functions $f_k$ are assumed to be time-invariant quadratic functions.  Matrix $F_{QP}$ and vector $f_{QP}$ in Equation \ref{Eq:ObjFunc} allow quadratic and linear costs to be incurred from the place and transition markings in both the engineering system net and operand nets. 
\begin{align}\label{Eq:ObjFunc}
Z &= \sum_{k=1}^{K-1} x^T[k] F_{QP} x[k] + f_{QP}^T x[k]
\end{align}
\begin{itemize}
\item $F_{QP}$ is a positive semi-definite, diagonal, quadratic coefficient matrix.
\item $f_{QP}$ is a linear coefficient matrix.
\end{itemize}

\vspace{0.5cm}
\subsubsection{Equality Constraints}

Matrix $A_{QP}$ and vector $B_{QP}$ in Equation \ref{Eq:EqualityConstraints} are constructed by concatenating constraints Equations  \ref{Eq:ESN-STF1}-\ref{CH6:eq:HFGTprog:comp:Fini}.  

\begin{align}\label{Eq:ESN-STF1}
-Q_{B}[k+1]+Q_{B}[k]+{M}^+U^+[k]\Delta T - {M}^-U^-[k]\Delta T=&0 && \!\!\!\!\!\!\!\!\!\!\!\!\!\!\!\!\!\!\!\!\!\!\!\!\!\!\!\!\!\!\!\!\!\!\!\!\!\!\!\!\!\forall k \in \{1, \dots, K\}\\  \label{Eq:ESN-STF2}
-Q_{\cal E}[k+1]+Q_{\cal E}[k]-U^+[k]\Delta T + U^-[k]\Delta T=&0 && \!\!\!\!\!\!\!\!\!\!\!\!\!\!\!\!\!\!\!\!\!\!\!\!\!\!\!\!\!\!\!\!\!\!\!\!\!\!\!\!\!\forall k \in \{1, \dots, K\}\\ \label{Eq:DurationConstraint}
 - U_\psi^+[k+k_{d\psi}]+ U_{\psi}^-[k] = &0 && \!\!\!\!\!\!\!\!\!\!\!\!\!\!\!\!\!\!\!\!\!\!\!\!\!\!\!\!\!\!\!\!\!\!\!\!\!\!\!\!\!\forall k\in \{1, \dots, K\} \quad \psi \in \{1, \dots, {\cal E}_S\}\\ \label{Eq:OperandNet-STF1}-Q_{Sl_i}[k+1]+Q_{Sl_i}[k]+{M}_{l_i}^+U_{l_i}^+[k]\Delta T - {M}_{l_i}^-U_{l_i}^-[k]\Delta T=&0 && \!\!\!\!\!\!\!\!\!\!\!\!\!\!\!\!\!\!\!\!\!\!\!\!\!\!\!\!\!\!\!\!\!\!\!\!\!\!\!\!\!\forall k \in \{1, \dots, K\} \quad i \in \{1, \dots, |L|\}\\ \label{Eq:OperandNet-STF2}
-Q_{{\cal E}l_i}[k+1]+Q_{{\cal E}l_i}[k]-U_{l_i}^+[k]\Delta T + U_{l_i}^-[k]\Delta T=&0 && \!\!\!\!\!\!\!\!\!\!\!\!\!\!\!\!\!\!\!\!\!\!\!\!\!\!\!\!\!\!\!\!\!\!\!\!\!\!\!\!\!\forall k \in \{1, \dots, K\} \quad i \in \{1, \dots, |L|\}\\ \label{Eq:OperandNetDurationConstraint}
- U_{xl_i}^+[k+k_{dxl_i}]+ U_{xl_i}^-[k] = &0 &&  \!\!\!\!\!\!\!\!\!\!\!\!\!\!\!\!\!\!\!\!\!\!\!\!\!\!\!\!\!\!\!\!\!\!\!\!\!\!\!\!\!
\forall k\in \{1, \dots, K\}, \: \forall x\in \{1, \dots, |{\cal E}_{l_i}\}|, \: l_i \in \{1, \dots, |L|\}\\ \label{Eq:SyncPlus}
U^+_L[k] - \widehat{\Lambda}^+ U^+[k] =&0 && \!\!\!\!\!\!\!\!\!\!\!\!\!\!\!\!\!\!\!\!\!\!\!\!\!\!\!\!\!\!\!\!\!\!\!\!\!\!\!\!\!\forall k \in \{1, \dots, K\}\\ \label{Eq:SyncMinus}
U^-_L[k] - \widehat{\Lambda}^- U^-[k] =&0 && \!\!\!\!\!\!\!\!\!\!\!\!\!\!\!\!\!\!\!\!\!\!\!\!\!\!\!\!\!\!\!\!\!\!\!\!\!\!\!\!\!\forall k \in \{1, \dots, K\}\\ \label{CH6:eq:HFGTprog:comp:Bound}
\begin{bmatrix}
D_{Up} & \mathbf{0} \\ \mathbf{0} & D_{Un}
\end{bmatrix} \begin{bmatrix}
U^+ \\ U^-
\end{bmatrix}[k] =& \begin{bmatrix}
C_{Up} \\ C_{Un}
\end{bmatrix}[k] && \!\!\!\!\!\!\!\!\!\!\!\!\!\!\!\!\!\!\!\!\!\!\!\!\!\!\!\!\!\!\!\!\!\!\!\!\!\!\!\!\!\forall k \in \{1, \dots, K\} \\\label{Eq:OperandRequirements}
\begin{bmatrix}
E_{Lp} & \mathbf{0} \\ \mathbf{0} & E_{Ln}
\end{bmatrix} \begin{bmatrix}
U^+_{l_i} \\ U^-_{l_i}
\end{bmatrix}[k] =& \begin{bmatrix}
F_{Lpi} \\ F_{Lni}
\end{bmatrix}[k] && \!\!\!\!\!\!\!\!\!\!\!\!\!\!\!\!\!\!\!\!\!\!\!\!\!\!\!\!\!\!\!\!\!\!\!\!\!\!\!\!\!\forall k \in \{1, \dots, K\}\quad i \in \{1, \dots, |L|\} \\\label{CH6:eq:HFGTprog:comp:Init} 
\begin{bmatrix} Q_B ; Q_{\cal E} ; Q_{SL} \end{bmatrix}[1] =& \begin{bmatrix} C_{B1} ; C_{{\cal E}1} ; C_{{SL}1} \end{bmatrix} \\ \label{CH6:eq:HFGTprog:comp:Fini}
\begin{bmatrix} Q_B ; Q_{\cal E} ; Q_{SL} ; U^- ; U_L^- \end{bmatrix}[K+1] =   &\begin{bmatrix} C_{BK} ; C_{{\cal E}K} ; C_{{SL}K} ; \mathbf{0} ; \mathbf{0} \end{bmatrix}
\end{align}

\begin{itemize}
\item Equations \ref{Eq:ESN-STF1} and \ref{Eq:ESN-STF2} describe the state transition function of an engineering system net (Defn \ref{Defn:ESN} \& \ref{Defn:ESN-STF}).
\item Equation \ref{Eq:DurationConstraint} is the engineering system net transition duration constraint where the end of the $\psi^{th}$ transition occurs $k_{d\psi}$ time steps after its beginning. 
\item Equations \ref{Eq:OperandNet-STF1} and \ref{Eq:OperandNet-STF2} describe the state transition function of each operand net ${\cal N}_{l_i}$ (Defn. \ref{Defn:OperandNet} \& \ref{Defn:OperandNet-STF}) associated with each operand $l_i \in L$.  
\item Equation \ref{Eq:OperandNetDurationConstraint} is the operand net transition duration constraint where the end of the $x^{th}$ transition occurs $k_{dx_{l_i}}$ time steps after its beginning. 
\item Equations \ref{Eq:SyncPlus} and \ref{Eq:SyncMinus} are synchronization constraints that couple the input and output firing vectors of the engineering system net to the input and output firing vectors of the operand nets respectively. $U_L^-$ and $U_L^+$ are the vertical concatenations of the input and output firing vectors $U_{l_i}^-$ and $U_{l_i}^+$, respectively.
\begin{align}
U_L^-[k]&=\left[U^-_{l_1}; \ldots; U^-_{l_{|L|}}\right][k] \\
U_L^+[k]&=\left[U^+_{l_1}; \ldots; U^+_{l_{|L|}}\right][k]
\end{align}
\item Equations \ref{CH6:eq:HFGTprog:comp:Bound} and \ref{Eq:OperandRequirements} are boundary conditions.  Eq. \ref{CH6:eq:HFGTprog:comp:Bound} is a boundary condition constraint that allows some of the engineering system net firing vectors decision variables to be set to an exogenous constant.  Eq. \ref{Eq:OperandRequirements} is a boundary condition constraint that allows some of the operand net firing vector decision variables to be set to an exogenous constant.  
\item Equations \ref{CH6:eq:HFGTprog:comp:Init} and \ref{CH6:eq:HFGTprog:comp:Fini} are the initial and final conditions of the engineering system net and the operand nets where $Q_{SL}$ is the vertical concatenation of the place marking vectors of the operand nets $Q_{Sl_i}$.
\begin{align}
Q_{SL}^-[k]&=\left[Q^-_{Sl_1}; \ldots; U^-_{Sl_{|L|}}\right][k] \\
U_{SL}^+[k]&=\left[U^+_{Sl_1}; \ldots; U^+_{Sl_{|L|}}\right][k]
\end{align}
\end{itemize}

\vspace{0.5cm}
\subsubsection{Inequality Constraints}
Matrix $D_{QP}$ and vector $E_{QP}$ in Equation \ref{ch6:eq:QPcanonicalform:3} place capacity constraints on the vector of decision variables at each time step $x[k] = \begin{bmatrix} Q_B ; Q_{\cal E} ; Q_{SL} ; Q_{{\cal E}L} ; U^- ; U^+ ; U^-_L ; U^+_L \end{bmatrix}[k] \quad \forall k \in \{1, \dots, K\}$. This flexible formulation allows capacity constraints on the place and transition markings in both the engineering system net and operand nets.  

\subsection{Electric Power System State Estimation}\label{Sec:WLSEPSSE}
In order to understand how the HFNMCF problem can be generalized to the WLSEHFGSE (in Sec. \ref{Sec:Methdology}), state estimation in electric power systems is taken as a well-known special case example.  The Weighted Least Squares Error Electric Power System State Estimator (WLSEPSSE) estimates the values of (algebraic) states in an electric power system in the presence of minimal measurement error.  More formally, 
\begin{align}\label{WLS_eq1} 
\text{minimize} \hspace{0.8cm} Z= e^TF_{QP}e\hspace{2.5cm}\\\label{WLS_eq2} 
s.t. = \begin{bmatrix}h_1(x_1,x_2,...,x_m)\\h_2(x_1,x_2,...,x_m)\\...\\h_m(x_1,x_2,...,x_m) \end{bmatrix} = \begin{bmatrix}C_{B1}\\C_{B2}\\...\\C_{Bm} \end{bmatrix} + \begin{bmatrix}e_1\\e_2\\...\\e_m \end{bmatrix}
\end{align}
where the objective function $Z$ is a quadratic function of measurement errors $e$, $h(x)$ is a collection of electric power system constraints, and $C_{B}$ is a vector of electric power system measurements\cite{Abur:2004:00}.  Typically, the measurements $C_B$ include active and reactive power injections $(P_i,Q_i)$ into a bus $i$, and active and reactive power flows $(P_{ij},Q_{ij})$ in a power line between buses $i$ and $j$.  Consequently, $h(x)$ consists of the power flow equations that define power injections into buses and flows within power lines\cite{Milano:2010:17}:

\begin{align}\label{Eq:ActivePower}
P_i&=|V_i|\sum_{j=1}^{n}|V_j|(G_{ij}\cos(\theta_i-\theta_j)+B_{ij}\sin(\theta_i-\theta_j)) \\
Q_i&=|V_i|\sum_{j=1}^{n}|V_j|(G_{ij}\sin(\theta_i-\theta_j)-B_{ij}\cos(\theta_i-\theta_j)) \\
P_{ij}&=|V_i||V_j|(G_{ij}\cos(\theta_i-\theta_j)+B_{ij}\sin(\theta_i-\theta_j))-G_{ij}|V_i|^2 \\
Q_{ij}&=|V_i||V_j|(G_{ij}\sin(\theta_i-\theta_j)-B_{ij}\cos(\theta_i-\theta_j))+B_{ij}|V_i|^2
\end{align}
where $G_{ij}$, $B_{ij}$ $\theta_{ij}$ are the conductance, the susceptance and the voltage angle difference between buses $i$ and $j$.  Alternatively, the measurements $C_B$ can include the real and imaginary components of generator currents ${\cal I}_{GR}$, ${\cal I}_{GI}$, the real and imaginary components of demanded currents ${\cal I}_{DR}$, ${\cal I}_{DI}$, the real and imaginary components of line currents ${\cal I}_{LR}$, ${\cal I}_{LI}$, and the real and imaginary components of the generator voltages ${V}_{GR}$, ${V}_{GI}$ and the real and imaginary components of bus voltages ${V}_{DR}$, ${V}_{DI}$.  Consequently, $h(x)$ consists of the Kirchoff's current law and Ohm's law in matrix form\cite{Chen:2004:00}.      
\begin{align}\label{Eq:NetworkFlowIVACOPF1}
\begin{bmatrix}
{\cal I}_{GR} \\
-{\cal I}_{DR} 
\end{bmatrix} - \begin{bmatrix}
A_G^T \\
A_D^T 
\end{bmatrix}{\cal I}_{LR} &= 0 \\\label{Eq:NetworkFlowReactivePower_IVACOPF_Final}
\begin{bmatrix}
{\cal I}_{GI} \\
-{\cal I}_{DI} 
\end{bmatrix} - \begin{bmatrix}
A_G^T \\
A_D^T 
\end{bmatrix}{\cal I}_{LI} &= 0 \\\label{Eq:NetworkFlowIVACOPF3}
\begin{bmatrix}
A_G^T \\
A_D^T 
\end{bmatrix}{\cal I}_{LR} -G\begin{bmatrix}
V_{GR} \\
V_{DR}
\end{bmatrix} + 
B\begin{bmatrix}
V_{GI} \\
V_{DI}
\end{bmatrix} &=0\\\label{Eq:NetworkFlowIVACOPF4}
\begin{bmatrix}
A_G^T \\
A_D^T 
\end{bmatrix}{\cal I}_{LI} - 
B\begin{bmatrix}
V_{GR} \\
V_{DR}
\end{bmatrix} - 
G\begin{bmatrix}
V_{GI} \\
V_{DI}
\end{bmatrix} &=0
\end{align}

\begin{thm}\label{Thm:WLSEEPSSE}
Under an assumption of \emph{zero} measurement error $e=0$, the Weighted Least Square Error Electric Power System State Estimation (WLSEEPSSE) problem is a special case of the Hetero-Functional Network Minimum Cost Flow (HFNMCF) problem.
\end{thm}

\begin{proof}
Equations \ref{Eq:ObjFunc1}-\ref{Eq:DeviceModels} are taken in turn.  
\begin{itemize}
\item For Equation \ref{Eq:ObjFunc1}, $Z=0=e^TF_{QP}e$.  Consequently, the HFNMCF objective function in Equation \ref{Eq:ObjFunc1} is equated to the WLSEEPSSE objective function in Equation \ref{WLS_eq1}.
\item For Equation \ref{Eq:EqualityConstraints}, the following conditions are imposed to reflect the special case of electric power system operation.  
\begin{itemize}[label={$\circ$}]
\item $K=1$ to reflect a single time step and steady-state conditions.
\item $Q_B[k]=0 \: \forall k$ to reflect that charge does not accumulate at electric power system buffers.
\item $k_{d\psi}=0 \: \forall k,\psi$ to reflect that the flow of electric current is instantaneous.  
\item $Q_{S}=\emptyset, U^-_{L}=\emptyset, U^+_{L}=\emptyset$ to reflect that electric power systems have no operand state.  
\end{itemize}
Consequently, 
\begin{itemize}[label={$\circ$}]
\item Equation \ref{Eq:ESN-STF1} can be equivalently recast by algebra as Equations \ref{Eq:NetworkFlowIVACOPF1} and \ref{Eq:NetworkFlowReactivePower_IVACOPF_Final}. 
\item Equations \ref{Eq:ESN-STF2} and \ref{Eq:DurationConstraint} collapse to triviality.  
\item Equations \ref{Eq:OperandNet-STF1}-\ref{Eq:SyncMinus} are not required.
\item Equation \ref{CH6:eq:HFGTprog:comp:Bound} is equivalent to \ref{WLS_eq2}.
\item Equation \ref{Eq:OperandRequirements} is not required.
\item Equations \ref{CH6:eq:HFGTprog:comp:Init}, and \ref{CH6:eq:HFGTprog:comp:Fini} collapse to triviality.
\end{itemize}
\item For Equation \ref{ch6:eq:QPcanonicalform:3}, the WLSEEPSSE does not impose capacity constraints.  Consequently, Equation \ref{ch6:eq:QPcanonicalform:3} is not required.
\item Finally, for Equation \ref{Eq:DeviceModels}, the auxiliary variables are set to electric power system voltages.  $Y=[V_{GR}; V_{GI}; V_{DR}; V_{DI}]$.  Consequently, Equation \ref{Eq:DeviceModels} is equivalent to Equations \ref{Eq:NetworkFlowIVACOPF3} and \ref{Eq:NetworkFlowIVACOPF4}. 
\end{itemize}
\end{proof}

Theorem \ref{Thm:WLSEEPSSE} provides significant motivational insight.  If the HFNMCF problem is revised to explicitly include measurement errors to exogenous data in Eq. \ref{CH6:eq:HFGTprog:comp:Bound} then it would become a formal generalization of the WLSEEPSSE problem.

\section{Weighted Least Squares Error Hetero-functional Graph State Estimation (WLSEHFGSE)} \label{Sec:Methdology}
Given the discussion of Sec. \ref{Sec:WLSEPSSE}, the WLSEHFGSE is presented as a formal generalization of the HFNMCF problem that admits measurement error of exogenous data in Eqs. \ref{Eq:ESNMeasurement} and \ref{Eq:SSNMeasurement}.  
\begin{align}
\text{minimize } \label{Eq:WLSEHFGSEOF} Z = \sum_{k=1}^{K-1} f_k(x[k],y[k]) & \\ \label{Eq:EqualityConstraintsOC}
\text{s.t. }-Q_{B}[k+1]+Q_{B}[k]+{M}^+U^+[k]\Delta T - {M}^-U^-[k]\Delta T=&0 && \!\!\!\!\!\!\!\!\!\!\!\!\!\!\!\!\!\!\!\!\!\!\!\!\!\!\!\!\!\!\!\!\!\!\!\!\!\!\!\!\! \forall k \in \{1, \dots, K\}\\ \label{Eq:TransitionConstraint2} 
-Q_{\cal E}[k+1]+Q_{\cal E}[k]-U^+[k]\Delta T + U^-[k]\Delta T=&0 && \!\!\!\!\!\!\!\!\!\!\!\!\!\!\!\!\!\!\!\!\!\!\!\!\!\!\!\!\!\!\!\!\!\!\!\!\!\!\!\!\!\forall k \in \{1, \dots, K\}\\ \label{Eq:DurationConstraint2}
 - U_\psi^+[k+k_{d\psi}]+ U_{\psi}^-[k] = &0 && \!\!\!\!\!\!\!\!\!\!\!\!\!\!\!\!\!\!\!\!\!\!\!\!\!\!\!\!\!\!\!\!\!\!\!\!\!\!\!\!\! \forall k \in \{1, \dots, K\} \quad \psi \in \{1, \dots, {\cal E}_S\}\\ \label{Eq:SSN-STF-QB2}
 -Q_{Sl_i}[k+1]+Q_{Sl_i}[k]+{M}_{l_i}^+U_{l_i}^+[k]\Delta T - {M}_{l_i}^-U_{l_i}^-[k]\Delta T=&0 && \!\!\!\!\!\!\!\!\!\!\!\!\!\!\!\!\!\!\!\!\!\!\!\!\!\!\!\!\!\!\!\!\!\!\!\!\!\!\!\!\!\forall k \in \{1, \dots, K\} \quad i \in \{1, \dots, |L|\}\\ \label{Eq:OperandNet-STF2OC}
-Q_{{\cal E}l_i}[k+1]+Q_{{\cal E}l_i}[k]-U_{l_i}^+[k]\Delta T + U_{l_i}^-[k]\Delta T=&0 && \!\!\!\!\!\!\!\!\!\!\!\!\!\!\!\!\!\!\!\!\!\!\!\!\!\!\!\!\!\!\!\!\!\!\!\!\!\!\!\!\!\forall k \in \{1, \dots, K\} \quad i \in \{1, \dots, |L|\}\\ 
- U_{xl_i}^+[k+k_{dxl_i}]+ U_{xl_i}^-[k] = &0 &&  \!\!\!\!\!\!\!\!\!\!\!\!\!\!\!\!\!\!\!\!\!\!\!\!\!\!\!\!\!\!\!\!\!\!\!\!\!\!\!\!\! \forall k\in \{1, \dots, K\}, \: \forall x\in \{1, \dots, |{\cal E}_{l_i}\}|, \: l_i \in \{1, \dots, |L|\}\\ 
U^+_L[k] - \widehat{\Lambda}^+ U^+[k] =&0 && \!\!\!\!\!\!\!\!\!\!\!\!\!\!\!\!\!\!\!\!\!\!\!\!\!\!\!\!\!\!\!\!\!\!\!\!\!\!\!\!\!\forall k \in \{1, \dots, K\}\\ \label{Eq:SyncMinusOC}
U^-_L[k] - \widehat{\Lambda}^- U^-[k] =&0 && \!\!\!\!\!\!\!\!\!\!\!\!\!\!\!\!\!\!\!\!\!\!\!\!\!\!\!\!\!\!\!\!\!\!\!\!\!\!\!\!\!\forall k \in \{1, \dots, K\}\\ \label{Eq:ESNMeasurement}
\begin{bmatrix}
D_{Up} & \mathbf{0} \\ \mathbf{0} & D_{Un}
\end{bmatrix} \begin{bmatrix}
U^+ \\ U^-
\end{bmatrix} =& \begin{bmatrix}
C_{Up} + {\cal E}_{Up}\\ C_{Un} + {\cal E}_{Un}
\end{bmatrix} &&  \\\label{Eq:SSNMeasurement}
\begin{bmatrix}
E_{Lp} & \mathbf{0} \\ \mathbf{0} & E_{Ln}
\end{bmatrix} \begin{bmatrix}
U^+_{l_i} \\ U^-_{l_i}
\end{bmatrix} =& \begin{bmatrix}
F_{Lpi} + {\cal E}_{Lpi} \\ F_{Lni} +  {\cal E}_{Lni}
\end{bmatrix} && \\\label{CH6:eq:HFGTprog:comp:InitOC} 
\begin{bmatrix} Q_B ; Q_{\cal E} ; Q_{SL} \end{bmatrix}[1] =& \begin{bmatrix} C_{B1} ; C_{{\cal E}1} ; C_{{SL}1} \end{bmatrix} \\ \label{CH6:eq:HFGTprog:comp:FiniOC}
\begin{bmatrix} Q_B ; Q_{\cal E} ; Q_{SL} ; U^- ; U_L^- \end{bmatrix}[K+1] =   &\begin{bmatrix} C_{BK} ; C_{{\cal E}K} ; C_{{SL}K} ; \mathbf{0} ; \mathbf{0} \end{bmatrix} \\
D_{CP}X \leq& E_{CP} \\ 
g(X,Y)=&0
\end{align}
where the vector of primary decision variables $X=\begin{bmatrix}x[1]; \ldots; x[K]\end{bmatrix}$ has been expanded to include the measurement errors of exogenous data in ${\cal E}_{Up}, {\cal E}_{Un}, {\cal E}_{Lp}, {\cal E}_{Ln}$.  
\begin{equation}
x[k] = \begin{bmatrix} Q_B ; Q_{\cal E} ; Q_{SL} ; Q_{{\cal E}L} ; U^- ; U^+ ; U^-_L ; U^+_L; {\cal E}_{Up}; {\cal E}_{Un}; {\cal E}_{Lp}; {\cal E}_{Ln}
\end{bmatrix}[k] \quad \forall k \in \{1, \dots, K\}
\end{equation}
Note that Equations \ref{Eq:ESNMeasurement} and \ref{Eq:SSNMeasurement} generalize Equations \ref{CH6:eq:HFGTprog:comp:Bound} and \ref{Eq:OperandRequirements} with the introduction of the measurement error variables ${\cal E}_{Up}, {\cal E}_{Un}, {\cal E}_{Lp}, {\cal E}_{Ln}$.  Furthermore, Equations \ref{Eq:ESNMeasurement} and \ref{Eq:SSNMeasurement} introduce the concatenated vectors 
$U^+=\begin{bmatrix}U^+[1]; \ldots ; U^+[K]\end{bmatrix}$, $U^-=\begin{bmatrix}U^-[1]; \ldots ; U^-[K]\end{bmatrix}$, $U_{l_i}^+=\begin{bmatrix}U_{l_i}^+[1]; \ldots ; U_{l_i}^+[K]\end{bmatrix}$, and
$U_{l_i}^-=\begin{bmatrix}U_{l_i}^-[1]; \ldots ; U_{l_i}^-[K]\end{bmatrix}$.  This allows these equations to introduce exogenous data that spans multiple time steps.  

\section{Application of WLSEHFGSE to the AMES}\label{Sec:WLSEHFGSE-AMES}
Given the exposition of the WLSEHFGSE problem in the previous section, this section turns to applying it to the AMES.  Sec. \ref{Sec:AMESCharacteristics} describes how the generic WLSEHFGSE collapses into an AMES-specific weighted least squares error state estimation optimization problem.  The remainder of this section elaborates on how each of the equations in this optimization problem is practically implemented.  

\subsection{Physical Characteristics of the AMES}\label{Sec:AMESCharacteristics}
It is important to recognize that the WLSEHFGSE is generic to a wide variety of complex systems-of-systems (much like the HFNMCF problem).  In contrast, specific systems-of-systems demonstrate intrinsic physical characteristics that collapse this general form to their specific application.  Such characteristics limit the feasible space of decision variables $x[k] = \begin{bmatrix} Q_B ; Q_{\cal E} ; Q_{SL} ; Q_{{\cal E}L} ; U^- ; U^+ ; U^-_L ; U^+_L; {\cal E}_{Up}; {\cal E}_{Un}; {\cal E}_{Lp}; {\cal E}_{Ln} \end{bmatrix}[k] \quad \forall k \in \{1, \dots, K\}$.   In the context of the AMES, these physical characteristics appear as the following modeling assumptions:  

\begin{itemize}
\item $Q_B[k+1]-Q_B[k]=0 \: \forall k \in \{1, \dots, K\}$:  Because the AMES reference architecture\cite{Thompson:2023:00} includes storage capabilities, the quantity of stored commodities can be tracked within $U^+$ and $U^-$.  The quantity of commodities in $Q_B[k]$ is therefore constant in time and can be removed from the optimization problem.
\item $k_{d\psi}=0 \: \forall k \in \{1, \dots, K\}, \psi \in \{1, \dots, {\cal E}_S\}$:  The AMES is assumed to demonstrate a continuous flow of energy commodities and therefore all transitions have zero duration.  Consequently,
\begin{itemize}[label=$\diamond$]
\item $U^+[k]=U^-[k] \: \forall k \in \{1, \dots, K\}$. Furthermore, the simplifying notation $U[k]$ is introduced and Eq. \ref{Eq:DurationConstraint2} is eliminated from the problem entirely.  Similarly, the simplifying notation $D_U$, $C_U$, and ${\cal E}_U$ is introduced and Eq. \ref{Eq:ESNMeasurement} is collapsed into a single block-row matrix equation.       
\item $Q_{\cal E}[k+1]-Q_{\cal E}[k]=0 \: \forall k \in \{1, \dots, K\}$:  The quantity of commodities in $Q_{\cal E}[k]$ is constant in time.  Both $Q_{\cal E}[k]$ and Eq. \ref{Eq:TransitionConstraint2} can be removed from the optimization problem.  
\end{itemize} 
\item $S_{l_i}=\emptyset \: \forall l_i \in L$:  Although the AMES reference architecture transforms commodities from one to another, none of the capabilities change the \emph{state} of these commodities.  Consequently, none of the energy commodities requires a state (to be tracked).  Furthermore, ${\cal N}_{l_i}=\emptyset \: \forall l_i \in L$, and Eq. \ref{Eq:SSN-STF-QB2}-\ref{Eq:SyncMinusOC}, \ref{Eq:SSNMeasurement} and \ref{CH6:eq:HFGTprog:comp:InitOC} can be removed from the optimization problem.
\item $g(X,Y)=0$.  As stated previously, the device model functions are optional and are not included in this application.  
\item Finally, the objective function depends on the square of the measurement error and the firing vectors.  $Z_{WLSE} = \sum_{k=1}^{K-1} x^T[k] F_{QP} x[k]$.  
\end{itemize}

As a result, the WLSEHFGSE, when applied to the AMES, appears as the following optimization problem:  
\begin{align}
\text{minimize } \label{Eq:WLSEOF} Z_{WLSE} = \sum_{k=1}^{K-1} x^T[k] F_{QP} x[k] & \\ \label{Eq:EqualityConstraintsOC2}
\text{s.t. }\bigg(M^+-{M}^-\bigg)U[k]\Delta T =&0 && \!\!\!\!\!\!\!\!\!\!\!\!\!\!\!\!\!\!\!\!\!\!\!\!\!\!\!\!\!\!\!\!\!\!\!\!\!\!\!\!\! \forall k \in \{1, \dots, K\}\\ \label{Eq:ESNMeasurement3}
D_{U}U =& C_{U} + {\cal E}_{U} && \\
U[K+1] =   &\mathbf{0} \\\label{Eq:ESNCapacity}
D_{CP}X \leq& E_{CP}
\end{align}
Furthermore, the decision variables become $x[k] = \begin{bmatrix} U ; {\cal E}_{U} \end{bmatrix}[k] \quad \forall k \in \{1, \dots, K\}$ to account for the flows through the engineering system net transitions $U$ and the measurement error ${\cal E}_{U}$ associated with the exogenous data $C_U$.  

\subsection{Weighted Least Squares Error Objective Function}
As stated previously, the weighted least square error objective function in Eq. \ref{Eq:WLSEOF} is a special case of Eq. \ref{Eq:WLSEHFGSEOF}.  It requires a positive semi-definite, diagonal, quadratic coefficient matrix $F_{QP}$.  $F_{QP}$ is composed of two diagonal matrices pertaining to the $U$ and ${\cal E}_U$ decision vectors respectively.    
\begin{align}\label{Eq:ObjWeight}
F_{QP} = \begin{bmatrix}
F_{U} & 0 \\
0 & F_{\cal E}
\end{bmatrix}
\end{align}
While it is possible to choose equal error weights ($F_{\cal E}=I$), such a choice does not take into account that the measured data $C_U$ can have very different scale and/or measurement units.  Consequently, a relatively small measurement error associated with a large measured value can have a disproportionate impact on the optimization.  Instead, $F_{\cal E}$ is chosen to be inversely proportional to the square of the maximum of each data time series.  
\begin{align}\label{Eq:ObjWeight2}
F_{\cal E} = diag\left(\frac{1}{max(C_U^2,2)}\right)
\end{align}
Similarly, it is possible to choose no firing vector weights ($F_{U}=0$).  However, such a choice does not take into consideration a situation where the WLSEHFGSE is provided with a data time series whose time step is coarse (e.g., annual) relative to the chosen time step (e.g., monthly).  In such a case, the WLSHFGSE will not distinguish between solutions that fulfill Eq. \ref{Eq:ESNMeasurement3} through non-zero values in a single time step and solutions that evenly distribute the non-zero values over time.  For example, the data time series associated with coal withdrawal and processed oil withdrawal are provided annually.  If $F_{U}=0$, then solutions that meet the annual coal and oil demand in \emph{a single month} are equally preferable to those that do so evenly over all twelve months.  To overcome this undesirable result, $F_{U}=0$ is assigned a small quadratic cost, which has the numerical effect of temporally distributing the solution to meet the coarse data time series.  
\begin{align}\label{Eq:ObjWeight1}
F_{U} = \alpha I
\end{align}
where $\alpha = 0.01\cdot \min(\max(C_U^2),2),1)$, $\max(\cdot,2)$ is a maximum function of the second (i.e. temporal) dimension and $\min(\cdot,1)$ is a minimum function of the first (i.e. feature) dimension.  The value of $0.01$ is chosen to ensure that the firing vector weights are approximately two orders of magnitude smaller than the data error weights.

\subsection{AMES Engineering System Net State Transition Function}
As stated previously, Eq. \ref{Eq:EqualityConstraintsOC2} is a special case of the engineering system net state transition function in Eq. \ref{Eq:EqualityConstraintsOC}.  It requires the calculation of the hetero-functional incidence matrix $M=M^+-M^-$ (Defn. \ref{Defn:ESN}).  As described previously in Sec. \ref{Introduction} and \ref{BackgroundCh7}, previous works have calculated the \emph{unweighted} engineering system net incidence matrix for the AMES\cite{Farid:2022:ISC-J49} using the HFGT toolbox\cite{Thompson:2023:07}.  The same toolbox may be used to calculate the \emph{weighted} engineering system net incidence matrix; but requires data on the process weights associated with each capability in the AMES' reference architecture\cite{Thompson:2023:00}.  A survey of the processes in the activity diagrams of the AMES' reference architecture reveals three types of processes.  
\begin{itemize}
\item \textbf{Injection/Withdrawal Processes:}   Some processes in the AMES reference architecture cross the system boundary.  They either describe a simple injection of operands into the system or a withdrawal of operands from the system.  In such cases, the corresponding value in the engineering system net incidence matrix is +1 or -1, respectively (and no additional data is needed).  
\item \textbf{Lossless Transportation Processes:} Other processes in the AMES reference architecture describe a simple (lossless) transportation of operands.  In such cases, there is both a +1 and -1 in the corresponding elements of the engineering system net incidence matrix (and no additional data is needed).  
\item \textbf{Transformation Processes:}  The remaining processes in the AMES reference architecture describe transformation processes where one operand is converted into another in some fixed ratio.  These transformation process weights are determined from EIA data sources \cite{EIA:2023:05,EIA:2023:06} and are summarized in Table \ref{weight-Table}.  More specifically, electric power generation processes are defined by average reported heat rates \cite{EIA:2023:05}.  Similarly, the weight assigned to the processing of crude oil to processed oil was identified by calculating the energy content of the average products from processing a single barrel of crude oil \cite{EIA:2023:06}.  
\end{itemize}

\begin{table}[h!]
\vspace{-0.2in}
\begin{center}
\caption{\label{weight-Table}Values of Processes Weights Used in the Hetero-functional Incidence Matrix}
\vspace{-0.1in}
\begin{tabular}{cccccc} 
\toprule
\textbf{Transformation Process} & \textbf{Input Operand} & \textbf{Input Weight(GJ)} & \textbf{Output Operand} & \textbf{Output Weight(GJ)} & \textbf{Source}\\
\hline
Generate Electric Power\newline from Processed Oil & Processed Oil & 3.289 & Electric Power & 1 & \cite{EIA:2023:05}\\ 
\hline
Generate Electric Power\newline from Processed Gas & Processed Gas & 2.253 & Electric Power & 1 & \cite{EIA:2023:05}\\ 
\hline
Generate Electric Power\newline from Syngas & Syngas & 2.253 & Electric Power & 1 & \cite{EIA:2023:05}\\ 
\hline
Generate Electric Power\newline from Coal & Coal & 3.102 & Electric Power & 1 & \cite{EIA:2023:05}\\ 
\hline
Generate Electric Power\newline from Nuclear & Nuclear & 3.056 & Electric Power & 1 & \cite{EIA:2023:05}\\ 
\hline
Process Crude Oil & Crude Oil & 1.285 & Processed Oil & 1 & \cite{EIA:2023:06}\\
\bottomrule
\end{tabular}
\end{center}
\vspace{-0.35in}
\end{table}

\subsection{AMES Engineering System Net Measurement Function}
As stated previously, Eq. \ref{Eq:ESNMeasurement3} is a special case of the engineering system net measurement function in Eq. \ref{Eq:ESNMeasurement}.  It requires the acquisition of the exogenous time-dependent behavioral data $C_U$.  Importantly, there is no limitation on the amount of data that can be included in $C_U$.  The data used in this case study is drawn from multiple EIA sources and is summarized in Table \ref{Constraint-Table}.  It describes each data set integrated into the optimization program, listing each of the following properties: Energy Modes, Data Set Description, Units, Temporal Granularity, Date Range, Spatial Granularity,  Corresponding Processes, Corresponding Resources, and EIA Source.  All five of the used datasets measure flows of energy commodities across the AMES' system boundary.  Four datasets describe energy withdrawals (or demands), while the fifth describes energy injections (of supply).
\begin{itemize}
\item The first integrated data set provides daily electric demand values for each balancing authority.  These values are summed to form monthly electric demand values for each region, which are then integrated into the WLSEHFGSE problem.  This dataset creates a constraint on the routed electric power that flows through the electric power transmission infrastructure, subsequently forcing electric power generation.  
\item The second integrated dataset imposes constraints on the amount of power each type of electric generation process produces.  This dataset provides target values for the amount of power each state generates for each fuel source.  
\item The third integrated dataset imposes a constraint on the demand for natural gas. This dataset provides target values for the monthly amount of natural gas withdrawn for each state.
\item The fourth integrated dataset imposes a constraint on the demand for processed oil. This dataset provides target values for the annual amount of processed oil withdrawn for each state.
\item The fifth integrated dataset imposes a constraint on the demand for withdrawing coal. This dataset provides target values for the annual amount of coal withdrawn for each state.
\end{itemize}
This demonstration of the WLSEHFGSE prioritizes data tied to the AMES' energy demand.  To satisfy the imposed demands, energy operand flows are routed through the AMES' capabilities to minimize the weighted least square measurement errors associated with the data in $C_U$.  As each element of data in $C_U$ is acquired, a corresponding element of measurement error in ${\cal E}_U$ is introduced.  Naturally, this case study can always be extended with additional, perhaps even redundant, datasets either from the EIA or other sources.   

\begin{table}[h!]
\vspace{-0.2in}
\begin{center}
\caption{\label{Constraint-Table}Exogeneous Time-Dependent Behavioral Data Sets Used in the AMES WLSEHFGSE Model.}\vspace{-0.10in}
\vspace{-0.1in}
\begin{tabular}{p{0.40in}p{1.1in}p{0.625in}p{0.50in}p{0.85in}p{0.9in}p{1.20in}p{0.25in}}\toprule
\textbf{Energy Modes} & \textbf{Data Set Description} & \textbf{Temporal Granularity} & \textbf{Date Range} & \textbf{Spatial \newline Granularity} &\textbf{Corresponding Processes} & \textbf{Corresponding Resources} & \textbf{EIA Source} \\\hline
Electric Power & Electric Power \newline Demand (MWH) & Daily & 2020.01.01 \newline 2021.12.31& Balancing \newline Authority Region & Route Electric \newline Power & Substations & \cite{EIA:2023:00}\\\hline
Electric Power & Electric Power Generation (1000 MWH) & Monthly & 2020.01.01 \newline 2021.12.31 &  States & Generate Electric Power from \{fuel source\}& Electric Power Plants & \cite{EIA:2023:01}\\\hline
Natural Gas & Natural Gas \newline Withdrawal (MMCF) & Monthly & 2020.01.01\newline 2021.12.31 &  States & Withdrawal \newline Natural Gas& Natural Gas Terminals \newline Natural Gas Receipt Delivery\newline Natural Gas Ports & \cite{EIA:2023:02}\\\hline
Processed Oil & Processed Oil \newline Withdrawal (BBTU) & Annual & 2020.01.01\newline 2021.12.31 &  States & Withdrawal \newline Processed Oil& Oil Terminals\newline Oil Receipt Delivery \newline Oil Ports & \cite{EIA:2023:03}\\ 
\hline
Coal & Coal \newline Withdrawal (BBTU) & Annual & 2020.01.01\newline 2021.12.31 & States & Withdrawal \newline Coal & Coal Terminals \newline Coal Ports & \cite{EIA:2023:04}\\\bottomrule
\end{tabular}
\end{center}
\vspace{-0.2in}
\end{table}

It is essential to acknowledge that the data presented in Table \ref{Constraint-Table} and incorporated into $C_U$ may have a spatial granularity different from the asset-level granularity defined in the engineering system's net transition flows $U$.  Similarly, this data may have a different temporal granularity than the engineering system net time step $\Delta T$.  In order to accommodate these disparate spatial and temporal granularities, the engineering system net transition flows $U$ may be reorganized into a two-dimensional matrix $\widetilde{U}$ of size $|{\cal E}_S| \times K$ such that $U=vec(\widetilde{U})$.  Similarly, the data contained in $C_U$ may be reorganized into a two-dimensional matrix $\widetilde{C}_U$ such that $C_U=vec(\widetilde{C}_U)$.   $\widetilde{C}_U$ is assumed to have $K_D$ columns corresponding to the number of time steps in the data and $|{\cal E}_D|$ rows corresponding to the number of capabilities (i.e. features) in the data.     

Once the exogenous time-dependent behavioral data $C_U$ is acquired, the coefficient matrix $D_U$ must be calculated.  Several situations are possible:
\begin{itemize}
\item Redundant Data with Relatively Fine Spatiotemporal Resolution:  In this case, $K_D>K$ or $|{\cal E}_D|>|{\cal E}_S|$ and the exogenous time-dependent behavior data $C_U$ must be downsampled to the same resolution as $U$ and the calculation of $D_U$ follows from one of the cases below.  
\item Complete Data with Same Spatio-temporal Resolution:  In this case, $K_D=K$ and $|{\cal E}_D|=|{\cal E}_S|$ and $D_U=I$.  
\item Redundant Data with Same Spatiotemporal Resolution:  In this case, $K_D=K$ and $|{\cal E}_D|>|{\cal E}_S|$.  $D_U$ has more rows than $U$ and is composed of the vertical concatenation of elementary basis row vectors with filled (one) elements that correspond to the associated element in $U$.  
\item Incomplete Data with Same Spatiotemporal Resolution:  In this case, $K_D=K$ and $|{\cal E}_D|<|{\cal E}_S|$.  $D_U$ has fewer rows than $U$ and is similarly composed of the vertical concatenation of elementary basis row vectors.
\item Incomplete Data with Relatively Coarse Spatiotemporal Resolution:  In this case, $D_U$ must be calculated from a capability aggregation matrix $D_{\cal E}$ and a temporal aggregation matrix $D_T$ as defined below.  
\end{itemize}

\begin{defn}[Capability Aggregation Matrix $D_{\cal E}$] The capability aggregation matrix $D_{\cal E} \in \{0,1\}^{|{\cal E}_S|\times |{\cal E}_D|}$ is a matrix whose element $D_{\cal E}(\psi_1,\psi_2)=1$ when the flow associated with system capability ${\epsilon}_{\psi_2} \in {\cal E}_S$ is part of the flow associated with data element $\widetilde{C}_U(\psi_1,k)$ at time k. 
\end{defn} 

\begin{defn}[Temporal Aggregation Matrix $D_{T}$] The temporal aggregation matrix $D_{T} \in \{0,1\}^{K\times K_D}$ is a matrix whose element $D_{T}(k_1,k_2)=1$ when the flows associated with time step $k_1 \in \{1, \ldots, K\}$ corresponds to the flows associated with $k_2 \in \{1, \ldots, K_D\}$ in the data $\widetilde{C}_U$. 
\end{defn} 
\noindent Consequently, the case of incomplete data with relatively coarse spatiotemporal resolution requires an Engineering System Net Measurement function with matricized decision variables:  
\begin{align}\label{Eq:ESNMeasurementFunc2}
D_{\cal E}\widetilde{U}D_{\cal T} = 
\widetilde{C}_{U} + \widetilde{\cal E}_{U}
\end{align}
Using the mixed product property\cite{Schacke:2004:00}, Eq. \ref{Eq:ESNMeasurementFunc2} simplifies straightforwardly to
\begin{align}\label{Eq:ESNMeasurementFunc3}
\begin{bmatrix}D_{\cal T}^T \otimes D_{\cal E}\end{bmatrix} U
= 
{C}_{U} + {\cal E}_{U}
\end{align}
where $\otimes$ is the Kronecker product.  Eq. \ref{Eq:ESNMeasurementFunc3} is equivalent to Eq. \ref{Eq:ESNMeasurement3} when $D_U=D_{\cal T}^T \otimes D_{\cal E}$.  In such a way, the WLSEHFGSE incorporates exogenous time-dependent behavioral data with minimal measurement error.   

\subsection{Capacity Constraints}

When applying the WLSEHFGSE problem to the AMES, not only must one consider constraints specific to the input data measurements, but also those that stem from the instantiated physical architecture.  This work demonstrates these constraints through the use of capacity constraints on the electric power generation processes.  Each instantiated electric power plant has a maximum electric power generation capacity that was Incorporated in the Platts Map Data Pro \cite{Platts:2017:00} datasets.  These capacities were used as the Platts datasets were used to guide the instantiation of the test cases \cite{Platts:2017:00, Thompson:2022:00}. These capacity constraints limit the amount of electric power that can be generated in a single time step by each capability associated with electric power generation by imposing a hard cap on the maximum amount.  This was imposed using the following equation

\begin{align}\label{Eq:CapConstraint}
\begin{bmatrix}
 A_{eg}
\end{bmatrix}
\begin{bmatrix}
 U[k]
\end{bmatrix} <=
\begin{bmatrix}
 C_{cap}
\end{bmatrix} && \forall k \in \{1, \dots, K\}
\end{align}

\noindent where $A_{eg}$ is a binary selector vector of size $|{\cal E}_S|\times 1$ with a one in every index corresponding to an electric power generation transition and a zero fro every other transition.  The vector $C_{cap}$ is also of size $|{\cal E}_S|\times 1$ with a zero in every transition index that does not correspond to electric power generation.  In the elements that do correspond to electric power generation, the transitions associated with the maximum capacity are stored.  In this way, no electric generation can exceed its maximum capacity at any time step.

\section{Visualized Results}\label{Sec:Results}
Given the exposition of the WLSEHFGSE problem in Sec. \ref{Sec:Methdology} and its application to the AMES in Sec. \ref{Sec:WLSEHFGSE-AMES}, this section provides the first integrated results of the AMES' open-source, asset-level granularity, behavioral model.  For comparison, the AMES is divided into two regions:  (1) Western Region from the Rocky Mountain states to the Pacific Ocean states and (2) Eastern Region from the Mississippi/Ohio River Valley states to the Atlantic Ocean states.  These are depicted in Fig. \ref{Fig:AMESRegions}.  Each region is composed of multiple regulatory regions (e.g., Independent System Operators, Balancing Authorities, etc) across its four constituent energy sectors.  

\begin{figure}[h!]\begin{center}\includegraphics[width=0.90\textwidth]{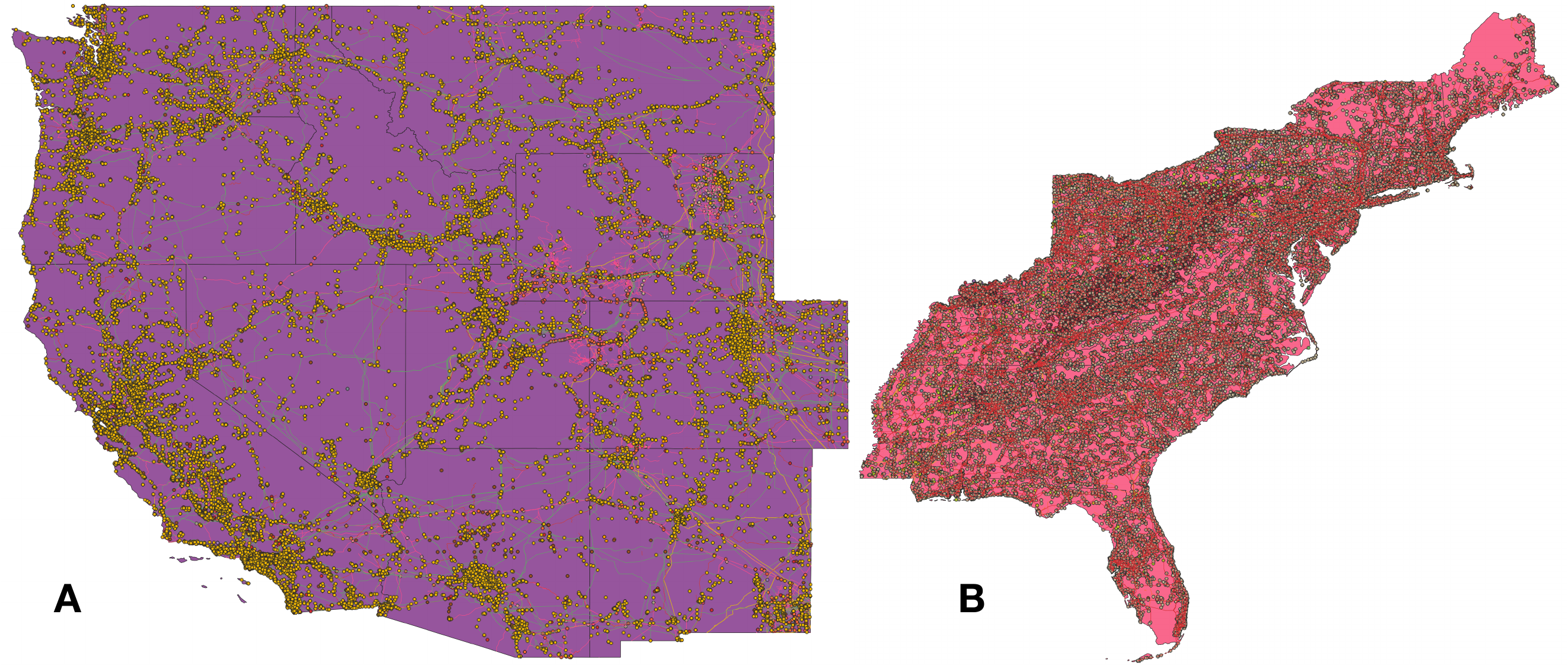}
\vspace{-0.1in}\caption{GIS Layers used the WLSEHFGSE applied to the AMES.  Subfigure (A) includes the Western Region from the Rocky Mountain states to the Pacific Ocean states.  Subfigure (b) includes the Eastern region from Mississippi/Ohio River Valley states to the Atlantic Ocean states.}\label{Fig:AMESRegions}\end{center}\end{figure}

The AMES's large size, asset-level granularity, and multi-faceted behavior necessitate a thoughtful approach to its visual presentation.  Sec. \ref{Sec:Sankey} transforms the activity diagrams of the AMES reference architecture \cite{Thompson:2023:00} into quantified Sankey diagrams.  These Sankey diagrams show how energy flows from one process to another.  Sec. \ref{Sec:ChoroplethDemand} uses choropleth (heat) maps to provide a spatial description of energy demand.  Sec. \ref{Sec:InterstateFlows} shows interstate network energy flows.  
Finally, Section \ref{Sec:ErrorAnalysis} presents a root-cause analysis of the measurement errors calculated by the WLSEHFGSE problem.  It uses bar graphs to show the relative size of measurement errors in the EIA data.  

\subsection{Sankey Diagram of Matter and Energy Flows}\label{Sec:Sankey}
The steady-state matter and energy conversion behavior of the AMES's Western Interconnect and Eastern Region is described with the animated Sankey diagrams shown in Figs.  \ref{Fig:SankeyWestern} and \ref{Fig:SankeyEastern} respectively.  Although the WLSEHFGSE problem is able to resolve decision variables associated with individual capabilities pertaining to the AMES' energy assets, the Sankey diagram aggregates decision variables with the same type of energy conversion or transportation process.  Consequently, the Sankey diagram does not provide insight into the AMES' spatial resolution, even when the WLSEHFGSE problem provides it.  Instead, the Sankey diagrams highlight the relative strength of process flows in the AMES.   Furthermore, the slider at the bottom of the plots operates over the duration of the date range specified in Table \ref{Constraint-Table}:  January 2020 - December 2021 with monthly time resolution.  Therefore, this analysis demonstrates how the relative strength of the process flows varies with the changing seasons.  Despite using this animated approach, Sankey diagrams still do not visualize capabilities that store matter and energy.  Ultimately, Figs.  \ref{Fig:SankeyWestern} and \ref{Fig:SankeyEastern} each show how the Sankey diagram behavior differs between a winter month (i.e. January 2020) and a summer month (i.e. July 2020).  

\begin{figure}[h!]\begin{center}\includegraphics[width=0.85\textwidth]{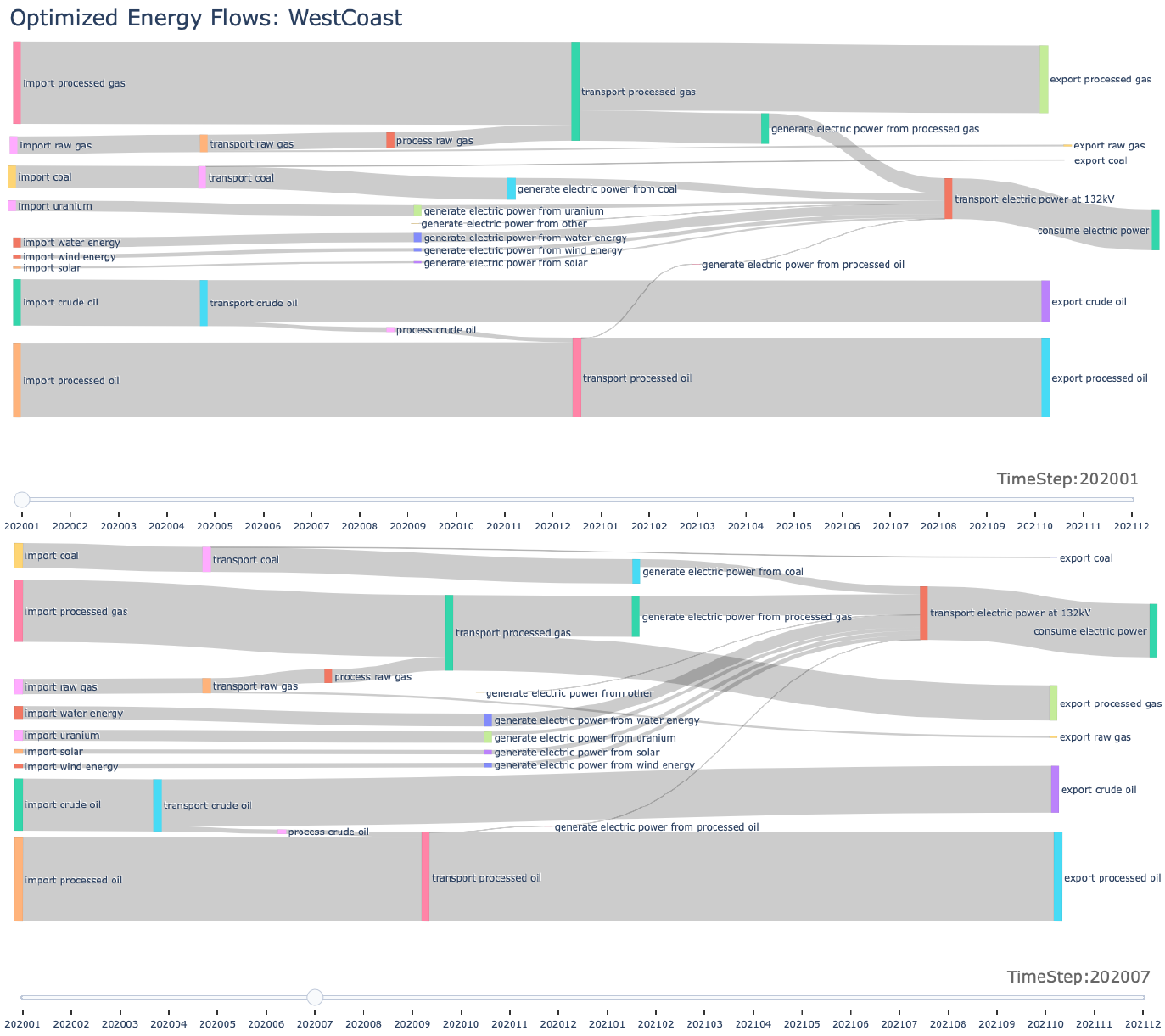}
\vspace{-0.1in}\caption{The Western Region Sankey Diagram shows the flows of mass/energy as they are transformed and transported from one process to another in (a.) January 2020 and (b.) July 2020.  A slider-bar has been included to animate the Sankey Diagram with a monthly time step.}\label{Fig:SankeyWestern}\end{center}\end{figure}

\begin{figure}[h!]\begin{center}\includegraphics[width=0.825\textwidth]{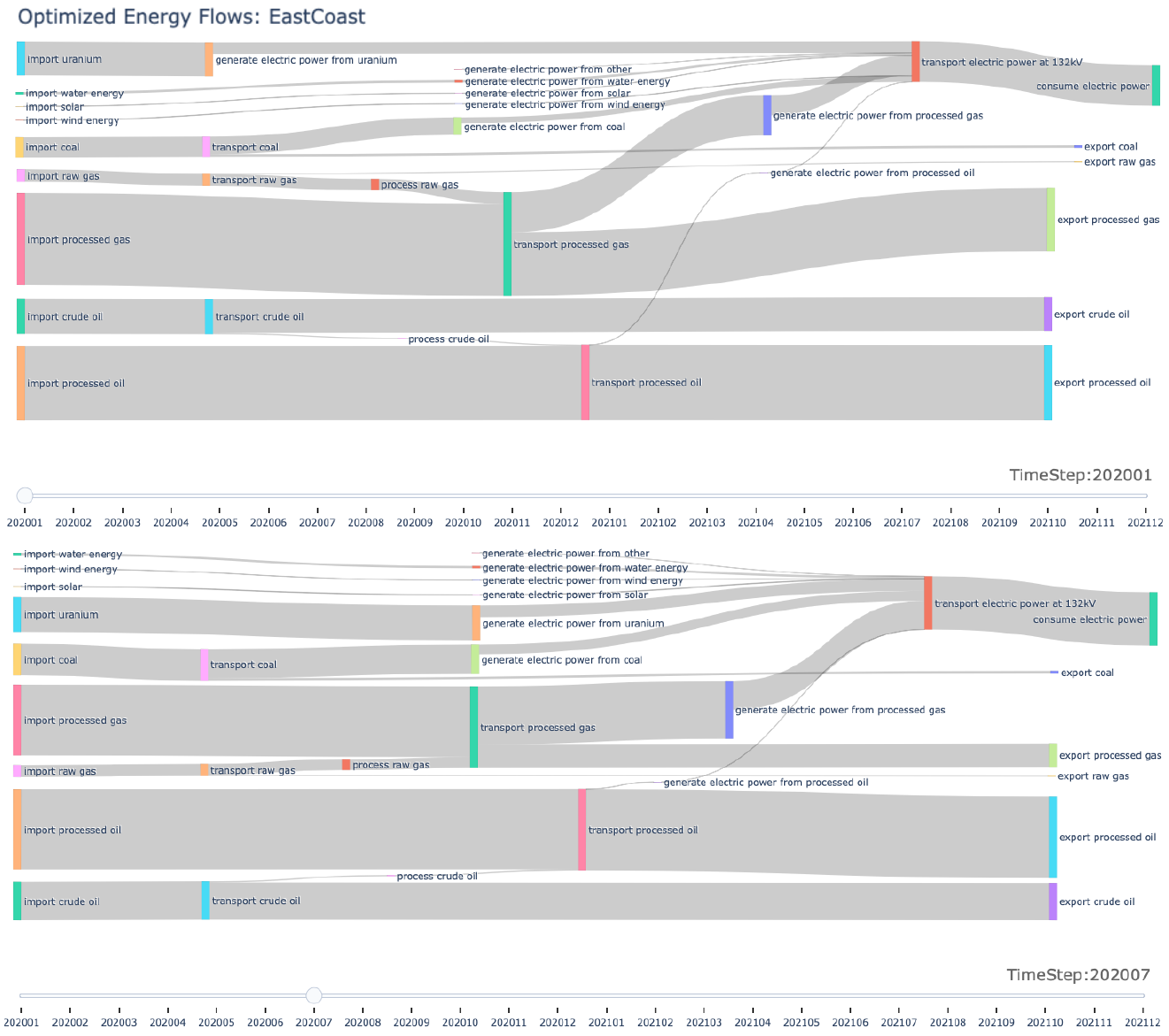}
\vspace{-0.1in}\caption{The Eastern Region Sankey Diagram shows the flows of mass/energy as they transformed and transported from one process to another.  A slider-bar has been included to animate the Sankey Diagram with a monthly time step.}\label{Fig:SankeyEastern}\end{center}\end{figure}

In all, Figs.  \ref{Fig:SankeyWestern} and \ref{Fig:SankeyEastern} illustrate the relative importance of the process flows for all four energy infrastructures in the AMES (i.e., coal, oil, natural gas, and electricity) as described in the reference architecture presented in the introduction.   As expected, oil and natural gas play a prominent role in both the Western Interconnect and the Eastern Region.  Meanwhile, the Western Interconnect relies heavily on processed (natural) gas, water energy (i.e., hydropower), and coal for electric power generation.  Similarly, the Eastern region relies heavily on processed (natural) gas, uranium (i.e., nuclear), and coal for electric power generation.  These results are consistent with previous works that investigated the structural properties of the AMES in differing regions\cite{Thompson:2022:00, Thompson:2022:ISC-C78}.  Much like the structural analysis, the behavioral analysis can also be used to investigate the impacts of geographical resources and political policies.  For example, it is worth noting that the Western interconnect has a greater reliance on solar electric power generation than the Eastern region, as their geographical landscape and policies favor solar as a renewable energy source.  Finally, it is worth noting that the use of processed gas increases during the winter and decreases during the warmer summer months in both regions.  This can be attributed to the use of processed gas for heating during the cold winters and the lack of heating demand during the summer.  Alternatively, electricity demand increases in the summer compared to winter, as the demand for air conditioning rises during the summer.  

\subsection{Choropleth Map of Spatial-Resolved Energy Demand}\label{Sec:ChoroplethDemand}
In contrast to the Sankey diagrams presented in the previous section, a spatial description of the relative energy demand intensities in the AMES' Western Interconnect and the Eastern Region is provided through the Choropleth Maps shown in Figs. \ref{WesternChoropleth} and \ref{EasternChoropleth}, respectively.  It is important to recognize that Figs. \ref{WesternChoropleth} and \ref{EasternChoropleth} plot values of the decision variables in the WLSEHFGSE problem rather than the exogenous input data found in $C_U$.  Naturally, these two quantities are only the same in the rare situation where measurement error is entirely absent.  Although the WLSEHFGSE problem can resolve the decision variables associated with individual capabilities pertaining to the AMES' energy assets, capability aggregation matrices have been used to depict the energy behavior at the chosen state level of granularity.  Additionally, slider bars have been added to distinguish the operand being consumed:  electricity, process (natural) gas, processed oil, and coal.  Again, as choropleth maps assume steady-state behavior, a second set of slider bars has been added to animate the choropleth maps over the January 2020 - December 2021 range with a monthly time step.  Consequently, this analysis shows how the relative intensities of energy demand change from season to season.  

\begin{figure}[h!]\begin{center}\includegraphics[width=0.8\textwidth]{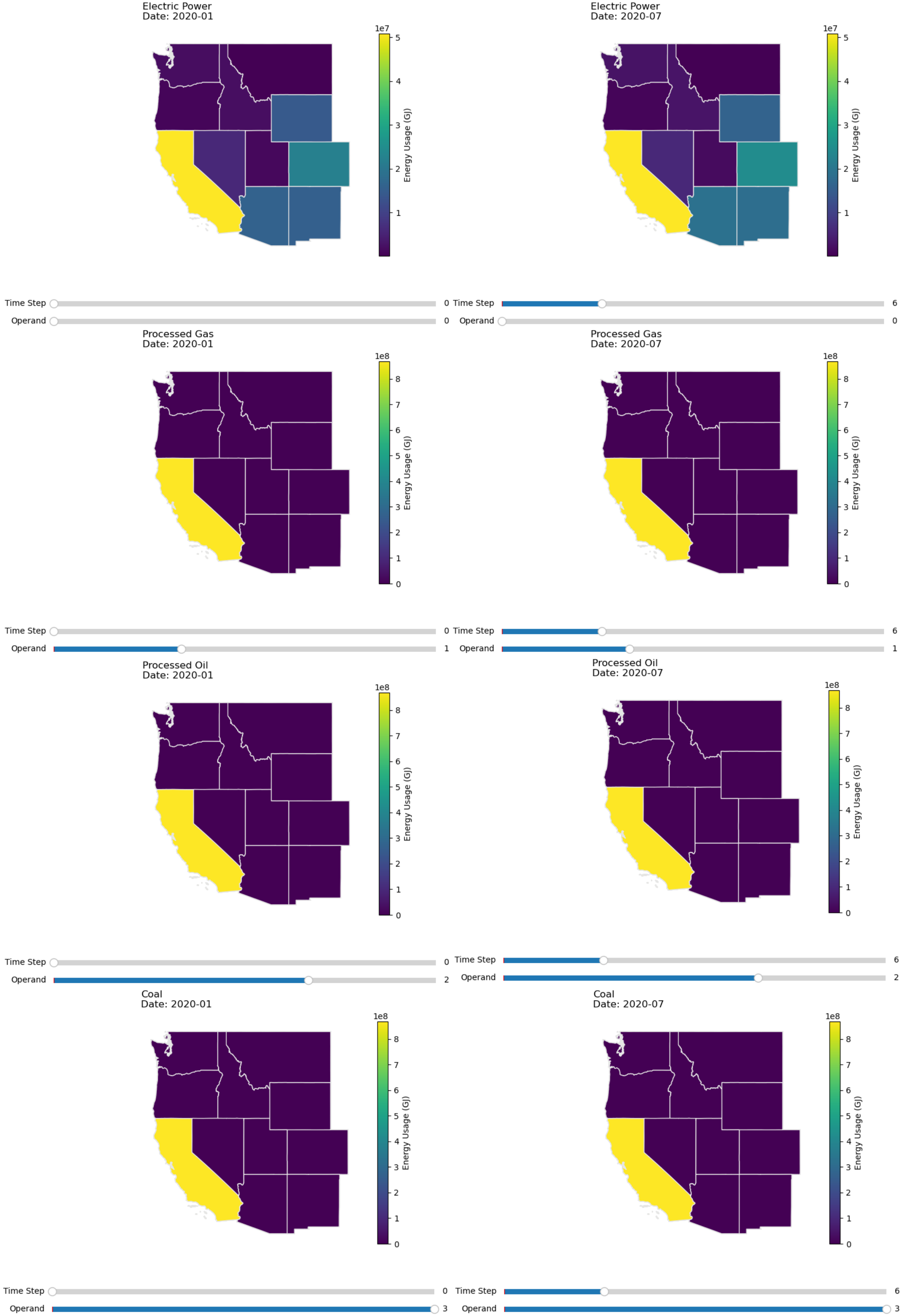}
\vspace{-0.1in}\caption{This figure presents a time series of the flows of mass and energy through the Western Interconnect using Choropleth maps.  The slider at the bottom of each snapshot labeled ``Time Step'' allows one to select which time step in the optimization they would like plotted.  This allows one to view how the flows change over time as seasons change.  The slider at the bottom of each snapshot labeled ``Operand'' allows one to select which operand in the optimization they would like plotted.}\label{WesternChoropleth}\end{center}\end{figure}

\begin{figure}[h!]\begin{center}\includegraphics[width=0.825\textwidth]{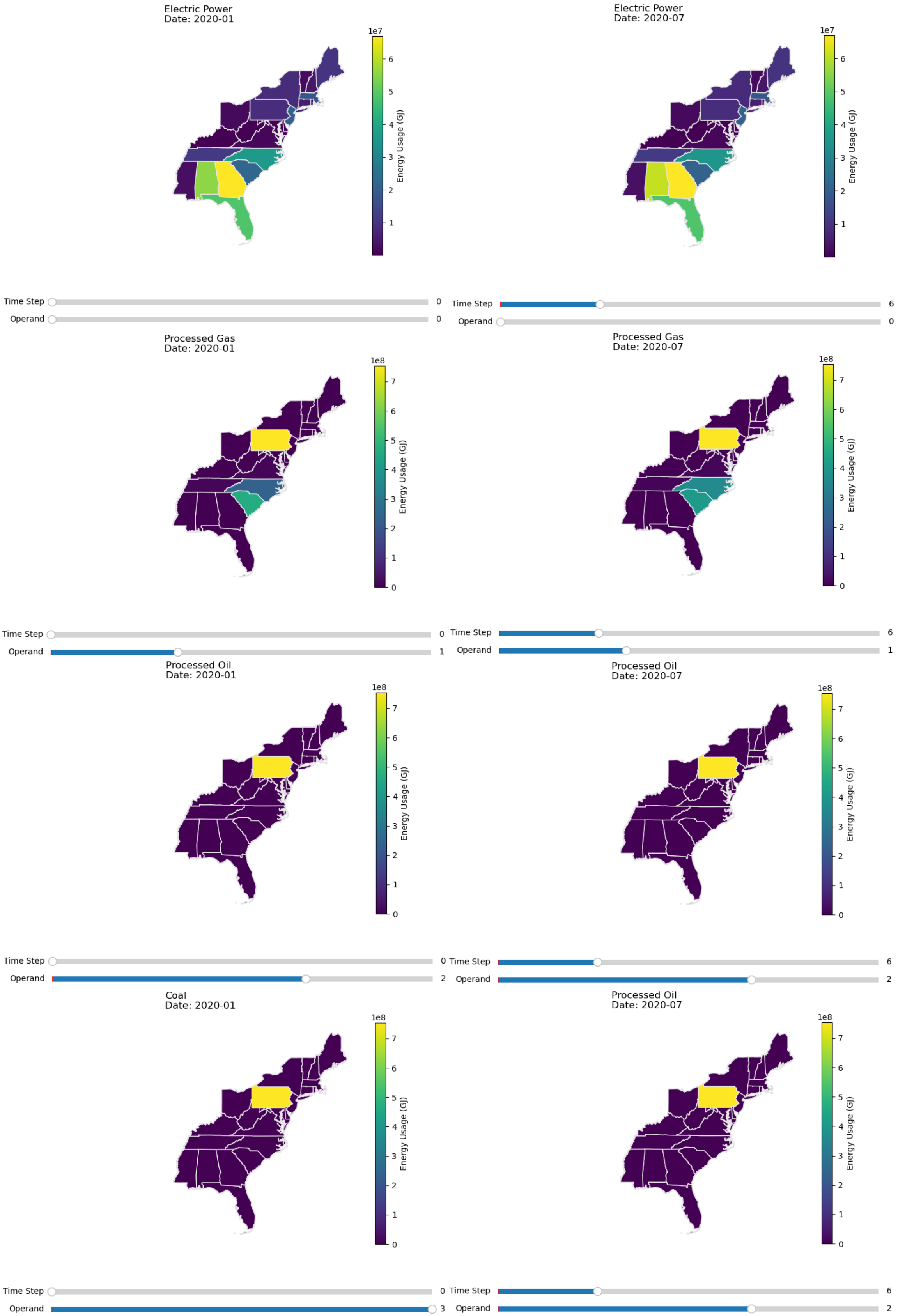}
\vspace{-0.1in}\caption{This figure presents a time series of the flows of mass and energy through the Eastern Region using Choropleth maps.  The slider at the bottom of each snapshot labeled ``Time Step'' allows one to select which time step in the optimization they would like plotted.  This allows one to view how the flows change over time as seasons change.  The slider at the bottom of each snapshot labeled ``Operand'' allows one to select which operand in the optimization they would like plotted.}\label{EasternChoropleth}\end{center}\end{figure}

In all, Figs. \ref{WesternChoropleth} and \ref{EasternChoropleth} show the relative intensities of energy demanded in all four energy infrastructures in the AMES.  As expected, California, as the most populous state in the Western Interconnect, has an energy demand across all commodities that dwarfs that of its neighbors.  Nevertheless, this predominance appears much more prominently for processed (natural) gas, processed oil, and coal than it does for electricity.  Meanwhile, in the Eastern Region, Pennsylvania shows its heavy reliance on processed (natural) gas, processed oil, and coal relative to its neighboring states.  Meanwhile, North Carolina, South Carolina, Georgia, Alabama, and Florida, with their relatively hot climates, exhibit a heavy reliance on electricity throughout the year.  

\subsection{Interstate Network Mass \& Energy Flows }\label{Sec:InterstateFlows}
\begin{figure}[h!]\begin{center}\includegraphics[width=0.6\textwidth]{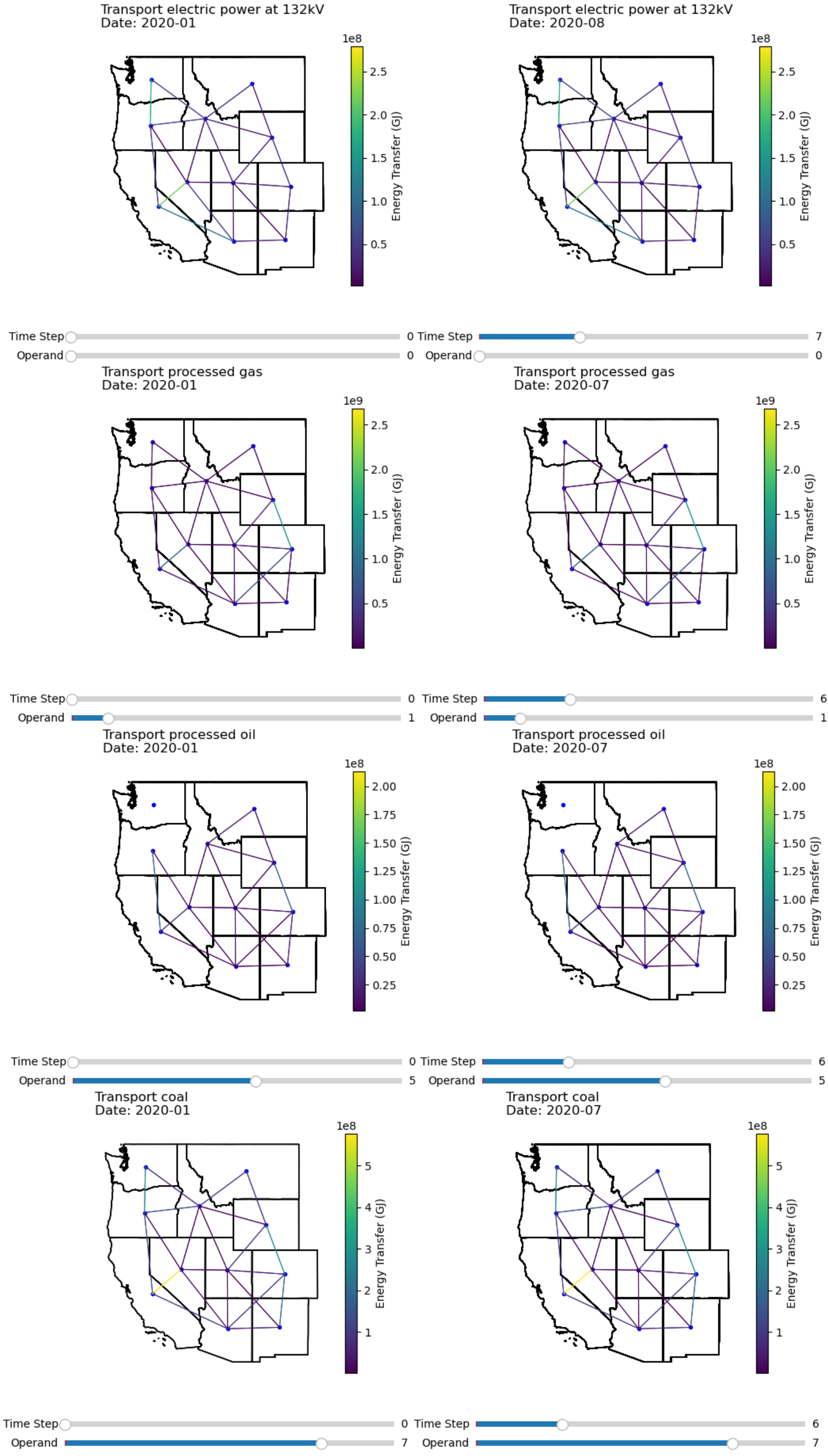}
\vspace{-0.1in}\caption{This figure presents a time series of the flows of mass and energy through the Western Interconnect between states. The slider at the bottom of each snapshot labeled ``Time Step'' allows one to select which time step in the optimization they would like plotted.  This allows one to view how the flows change over time as seasons change.  The slider at the bottom of each snapshot labeled ``Operand'' allows one to select which operand in the optimization they would like plotted.}\label{WesternInterState}\end{center}\end{figure}

\begin{figure}[h!]\begin{center}\includegraphics[width=0.60\textwidth]{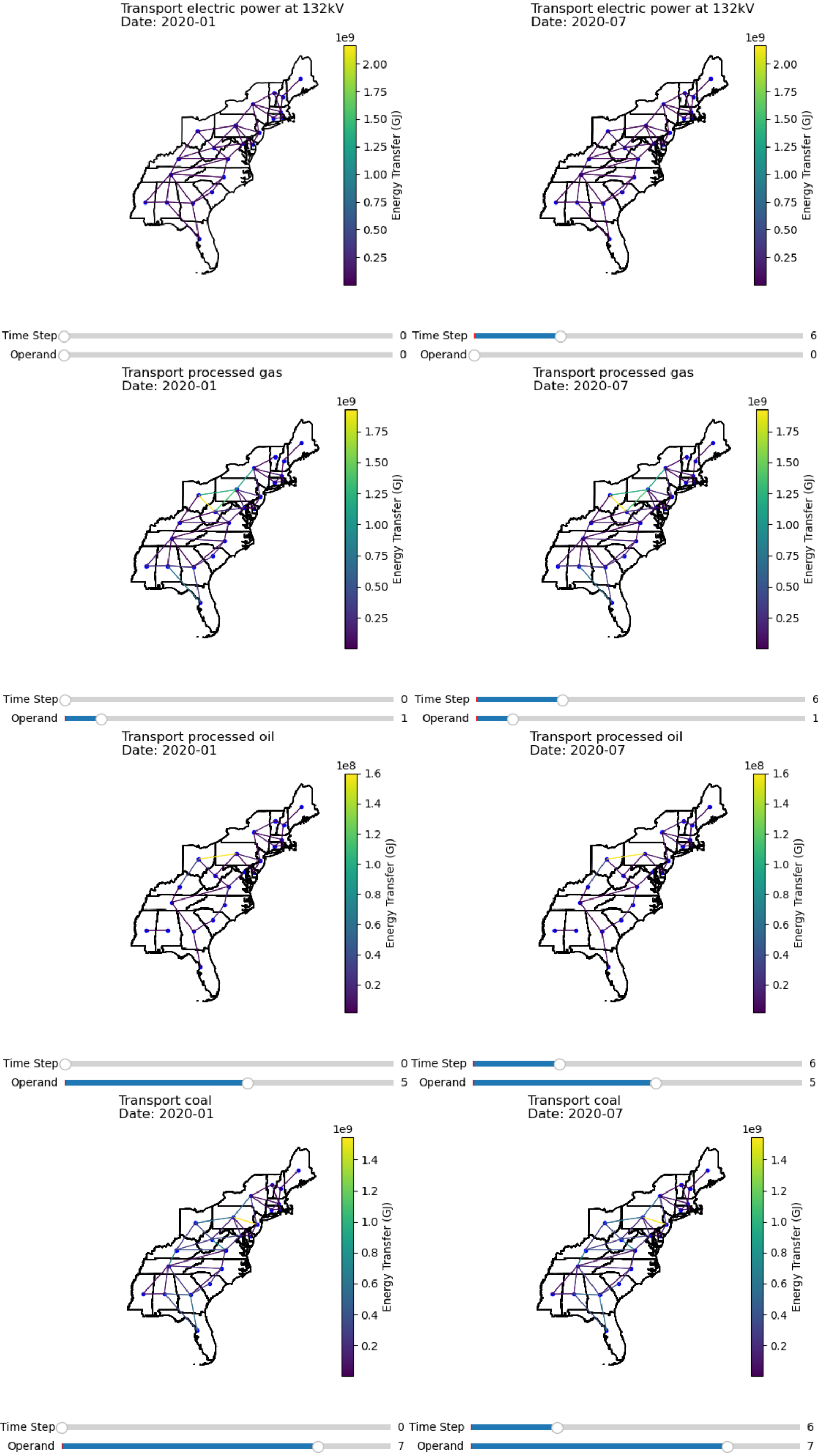}
\vspace{-0.1in}\caption{This figure presents a time series of the flows of mass and energy through the Eastern Region between states. The slider at the bottom of each snapshot labeled ``Time Step'' allows one to select which time step in the optimization they would like plotted.  This allows one to view how the flows change over time as seasons change.  The slider at the bottom of each snapshot labeled ``Operand'' allows one to select which operand in the optimization they would like plotted.}\label{EasternInterState}\end{center}\end{figure}

In addition to the Sankey diagrams and choropleth maps, Figs. \ref{WesternInterState} and \ref{EasternInterState} show the interstate network mass and energy flows in the Western Interconnect and the Eastern Region, respectively.  While the choropleth maps provide a spatial description of energy consumption, they ultimately assume that each region is independent.  In contrast, Figs. \ref{WesternInterState} and \ref{EasternInterState} emphasize the mass and energy flow between states.  Although the WLSEHFGSE problem can resolve the decision variables associated with individual capabilities pertaining to the AMES' energy assets, capability aggregation matrices have once again been used to depict the energy behavior at the chosen interstate level of granularity.  Additionally, slider bars have been added to distinguish the operand being transported:  electricity, processed (natural) gas, processed oil, and coal.  Again, as network flow diagrams present a static depiction of interstate flows, a second set of slider bars has been added to animate the network flows over the January 2020 - December 2021 range with a monthly time step.  Unlike the choropleth maps of the previous subsection, this visualization does not show a strong variation in interstate flows.  To the contrary, both Figs. \ref{WesternInterState} and \ref{EasternInterState} show that both the AMES' Western Interconnect and Eastern Region are particularly well-connected in their interstate flows of coal, oil, natural gas, and electricity.  

\subsection{Analysis of Error Sources in the WLSEHFGSE}\label{Sec:ErrorAnalysis}
In addition to the AMES' spatio-temporal behavior described in the previous section, it is also useful to analyze the sources of error in the WLSEHFGSE.  The weighted least square error objective function in Eq. \ref{Eq:WLSEOF} introduces an error component for each measurement error decision variable in Eq. \ref{Eq:ESNMeasurement3}.  Consequently, it is important to understand which input data values contribute to the accrual of error in the objective function.  Figs. \ref{WesternErrorPlots} and \ref{EasternErrorPlots} quantify these measurement errors relative to the input data values to which they pertain.  While it is possible to conduct an analysis of error sources at the level of each input data value, such an approach does not immediately reveal the types of errors that have accrued.  In Figs. \ref{WesternErrorPlots} and \ref{EasternErrorPlots} utilize capability aggregation and temporal aggregation matrices to aggregate errors originating from the same process in the same state.  Consequently, the y-axis shows the magnitude of error associated with each type of constraint on the x-axis.  Additionally, the x-axis has been sorted from greatest error on the left to least error on the right.  

\begin{figure*}[htb!]
\centering
\includegraphics[width=6in]{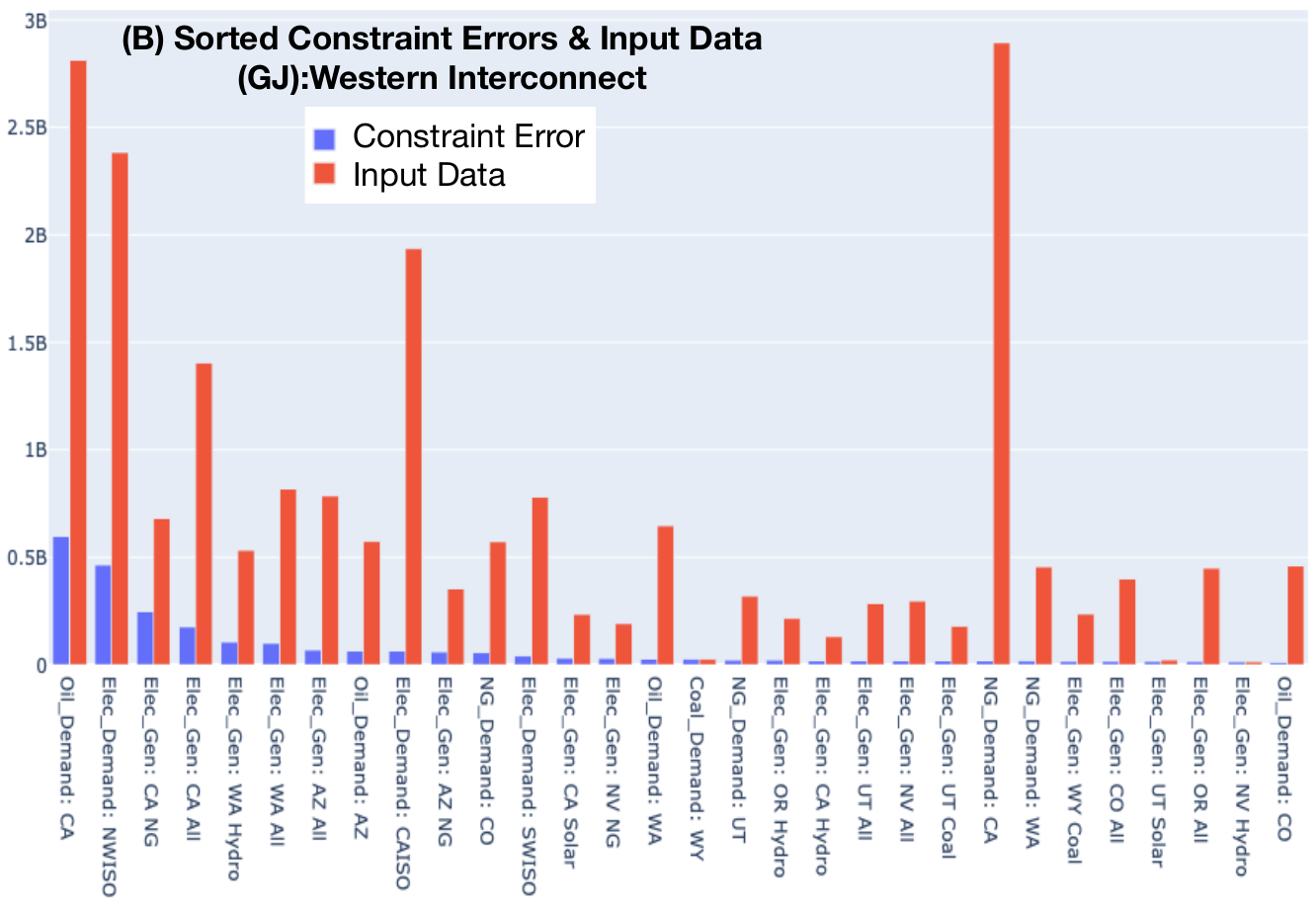}
\vspace{-0.1in}
\caption{This figure presents the error terms for the Western Interconnect case. The error contributed by each constraint on the WLSEHFGSE problem is shown next to the corresponding input data's imposed flow. }\label{WesternErrorPlots}
\vspace{-0.1in}
\end{figure*}

In the Western Interconnect, the greatest error (in blue) originates from the imposed Oil demand in CA, followed by the amount of electric demand in the NWISO, and then the electric generation in CA.  Interestingly, these measurement errors accrue because the capacity values in Eq. \ref{Eq:ESNCapacity} limit the degree to which AMES' capabilities can meet the measurement functions in Eq. \ref{Eq:ESNMeasurement3}.  Had these capacity values been omitted, as is done in the electric power system state estimator (discussed in Sec. \ref{Sec:WLSEPSSE}), many of these errors could have been reduced if not entirely eliminated.  Instead, WLSEHFGSE reveals that the measured behavioral data is not entirely consistent with the data pertaining to the AMES' structural capacities.  Nevertheless, Fig. \ref{WesternErrorPlots} shows that the error accrued in the Western Interconnect is modest relative to the magnitude of the input data values.  

\begin{figure*}[htb!]
\centering
\includegraphics[width=6in]{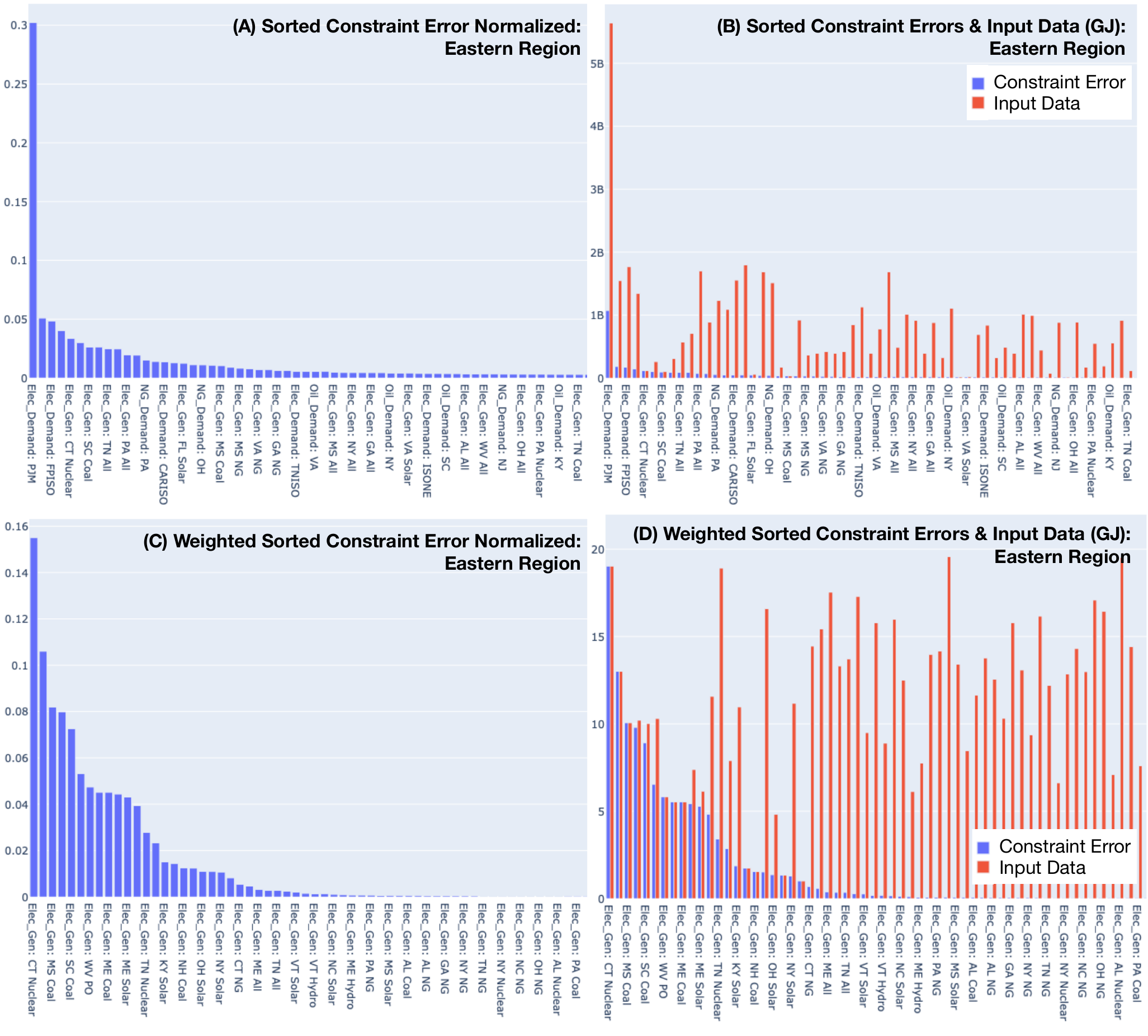}
\vspace{-0.1in}
\caption{This figure presents the error terms for the Eastern Region. The error contributed by each constraint on the WLSEHFGSE problem is shown next to the corresponding input data's imposed flow.}\label{EasternErrorPlots}
\vspace{-0.1in}
\end{figure*}

The same error trends can be seen in the Eastern Region error plot shown in Figure \ref{EasternErrorPlots}.  It is worth noting that the three most significant sources of error in the Eastern region are those corresponding to the electric power demand in PJM and FPISO, as well as the electric power generation from nuclear power in CT. Interestingly, the withdrawals of oil, natural gas, and coal do not incur errors in this optimal solution of WLSEHFGSE.  This is because there is no constraint on the total amount of mass or energy a resource can import or withdraw for these operands.  As a result, the only ways one of these input data sets could incur error would be if there were no resources to perform the import or withdrawal function, or if additional redundant but conflicting input data were introduced.  In contrast, the amount electric power generated by a power plant is constrained by its capacity value in the input data.  As a result, power generation facilities in the Eastern Region accrue errors like their counterparts in the Western Interconnect.

Finally, it is essential to acknowledge that errors in the WLSEHFGSE often originate from conflicting and redundant datasets.   Table \ref{Error-Table} shows several of the AMES' input data sets that include redundant and conflicting datasets.  Two broad categories appear immediately.  First, the imposed electric power generation from all sources and the sum of electric power generation from specific fuel sources do not match.  Second, the electric power demand for the independent system operators does not match the total generation or the summed electric power generation from specified fuel sources within those regions.  Since the demand and imposed electric power generation do not match, an error must be introduced to fill the difference.  Despite these errors, the WLSEHFGSE provides the best estimate of the AMES' state.  For example, the data discrepancy between the electric power generation values and the electric power demand values can be shared between the associated error terms.  By distributing the accumulation of error, the WLSEHFGSE problem minimizes the total error, reconciles data inconsistencies, and matches the resulting flows to the imposed input data as best as possible.  

\newpage
\scriptsize
\begin{longtable}{|c|c|c|c|c|c|}
\caption{AMES' Redundant and Conflicting Input Data}\label{Error-Table} \\
\hline
\textbf{Input Data 1} & \textbf{Input Data 2} & \textbf{Data 1 Total (GJ)} & \textbf{Data 2 Total (GJ)} & \textbf{Total Difference (GJ)} & \textbf{Sources} \\
\hline
\endfirsthead

\hline
\textbf{Input Data 1} & \textbf{Input Data 2} & \textbf{Data 1 Total (GJ)} & \textbf{Data 2 Total (GJ)} & \textbf{Total Difference (GJ)} & \textbf{Sources} \\
\hline
\endhead

\hline
\endfoot

\hline
\endlastfoot

\makecell{Generate Electric Power\\from All Sources: NY} & \makecell{Generate Electric Power\\from Summed Sources: NY} & 915,084,000 & 914,774,400 & 309,600 & \cite{EIA:2023:01}, \cite{EIA:2023:01}\\ 
\hline
\makecell{Generate Electric Power\\from All Sources: ME} & \makecell{Generate Electric Power\\from Summed Sources: ME} & 75,229,200 & 74,883,600 & 345,600 & \cite{EIA:2023:01}, \cite{EIA:2023:01}\\ 
\hline
\makecell{Generate Electric Power\\from All Sources: NH} & \makecell{Generate Electric Power\\from Summed Sources: NH} & 120,715,200 & 120,351,600 & 363,600 & \cite{EIA:2023:01}, \cite{EIA:2023:01}\\ 
\hline
\makecell{Generate Electric Power\\from All Sources: VT} & \makecell{Generate Electric Power\\from Summed Sources: VT} & 15,307,200 & 15,094,800 & 212,400 & \cite{EIA:2023:01}, \cite{EIA:2023:01}\\ 
\hline
\makecell{Generate Electric Power\\from All Sources: MA} & \makecell{Generate Electric Power\\from Summed Sources: MA} & 135,640,800 & 135,381,600 & 259,200 & \cite{EIA:2023:01}, \cite{EIA:2023:01}\\ 
\hline
\makecell{Generate Electric Power\\from All Sources: CT} & \makecell{Generate Electric Power\\from Summed Sources: CT} & 306,932,400 & 306,615,600 & 316,800 & \cite{EIA:2023:01}, \cite{EIA:2023:01}\\
\hline
\makecell{Generate Electric Power\\from All Sources: RI} & \makecell{Generate Electric Power\\from Summed Sources: RI} & 65,538,000 & 65,350,800 & 187,200 & \cite{EIA:2023:01}, \cite{EIA:2023:01}\\
\hline
\makecell{Generate Electric Power\\from All Sources: NJ} & \makecell{Generate Electric Power\\from Summed Sources: NJ} & 441,108,000 & 43,959,600 & 1,512,000 & \cite{EIA:2023:01}, \cite{EIA:2023:01}\\
\hline
\makecell{Generate Electric Power\\from All Sources: NC} & \makecell{Generate Electric Power\\from Summed Sources: NC} & 915,393,600 & 915,008,400 & 385,200 & \cite{EIA:2023:01}, \cite{EIA:2023:01}\\
\hline
\makecell{Generate Electric Power\\from All Sources: PA} & \makecell{Generate Electric Power\\from Summed Sources: PA} & 1,697,270,400 & 1,692,860,400 & 4,410,000 & \cite{EIA:2023:01}, \cite{EIA:2023:01}\\
\hline
\makecell{Generate Electric Power\\from All Sources: AL} & \makecell{Generate Electric Power\\from Summed Sources: AL} & 1,008,946,800 & 1,008,644,400 & 302,400 & \cite{EIA:2023:01}, \cite{EIA:2023:01}\\
\hline
\makecell{Generate Electric Power\\from All Sources: FL} & \makecell{Generate Electric Power\\from Summed Sources: FL} & 1,795,406,400 & 1,786,885,200 & 8,521,200 & \cite{EIA:2023:01}, \cite{EIA:2023:01}\\
\hline
\makecell{Generate Electric Power\\from All Sources: GA} & \makecell{Generate Electric Power\\from Summed Sources: GA} & 879,530,400 & 878,364,000 & 1,166,400 & \cite{EIA:2023:01}, \cite{EIA:2023:01}\\
\hline
\makecell{Generate Electric Power\\from All Sources: MD} & \makecell{Generate Electric Power\\from Summed Sources: MD} & 267,314,400 & 266,907,600 & 406,800 & \cite{EIA:2023:01}, \cite{EIA:2023:01}\\
\hline
\makecell{Generate Electric Power\\from All Sources: MS} & \makecell{Generate Electric Power\\from Summed Sources: MS} & 483,462,000 & 483,199,200 & 262,800 & \cite{EIA:2023:01}, \cite{EIA:2023:01}\\
\hline
\makecell{Generate Electric Power\\from All Sources: WV} & \makecell{Generate Electric Power\\from Summed Sources: WV} & 440,946,000 & 440,402,400 & 543,600 & \cite{EIA:2023:01}, \cite{EIA:2023:01}\\
\hline
\makecell{Generate Electric Power\\from All Sources: DE} & \makecell{Generate Electric Power\\from Summed Sources: DE} & 34,196,400 & 33,004,800 & 1,191,600 & \cite{EIA:2023:01}, \cite{EIA:2023:01}\\
\hline
\makecell{Generate Electric Power\\from All Sources: KY} & \makecell{Generate Electric Power\\from Summed Sources: KY} & 480,366,000 & 480,067,200 & 298,800 & \cite{EIA:2023:01}, \cite{EIA:2023:01}\\
\hline
\makecell{Generate Electric Power\\from All Sources: OH} & \makecell{Generate Electric Power\\from Summed Sources: OH} & 888,944,400 & 875,818,800 & 13,125,600 & \cite{EIA:2023:01}, \cite{EIA:2023:01}\\
\hline
\makecell{Generate Electric Power\\from All Sources: TN} & \makecell{Generate Electric Power\\from Summed Sources: TN} & 574,599,600 & 574,185,600 & 414,000 & \cite{EIA:2023:01}, \cite{EIA:2023:01}\\
\hline
\makecell{Generate Electric Power\\from All Sources: VA} & \makecell{Generate Electric Power\\from Summed Sources: VA} & 707,479,200 & 707,162,400 & 316,800 & \cite{EIA:2023:01}, \cite{EIA:2023:01}\\
\hline
\makecell{Generate Electric Power\\from All Sources: SC} & \makecell{Generate Electric Power\\from Summed Sources: SC} & 708,868,800 & 708,544,800 & 324,000 & \cite{EIA:2023:01}, \cite{EIA:2023:01}\\
\hline
\makecell{Generate Electric Power\\from All Sources: NY} & \makecell{Electric Power\\Demand: NYISO} & 915,084,000 & 1,087,836,447.6 & 172,752,447.6 & \cite{EIA:2023:01}, \cite{EIA:2023:00}\\
\hline
\makecell{Generate Electric Power\\from Summed Sources: NY} & \makecell{Electric Power\\Demand: NYISO} & 914,774,400 & 1,087,836,447.6 & 173,062,047.6 & \cite{EIA:2023:01}, \cite{EIA:2023:00}\\
\hline
\makecell{Generate Electric Power\\from All Sources: NE} & \makecell{Electric Power\\Demand: ISONE} & 719,362,800 & 837,803,440.8 & 118,440,640.8 & \cite{EIA:2023:01}, \cite{EIA:2023:00}\\
\hline
\makecell{Generate Electric Power\\from Summed Sources: NE} & \makecell{Electric Power\\Demand: ISONE} & 717,678,000 & 837,803,440.8 & 120,125,440.8 & \cite{EIA:2023:01}, \cite{EIA:2023:00}\\
\hline
\makecell{Generate Electric Power\\from All Sources: PJM} & \makecell{Electric Power\\Demand: PJM} & 4,957,624,800 & 5,640,611,054.4 & 682,986,254.4 & \cite{EIA:2023:01}, \cite{EIA:2023:00}\\
\hline
\makecell{Generate Electric Power\\from Summed Sources: PJM} & \makecell{Electric Power\\Demand: PJM} & 4,935,819,600 & 5,640,611,054.4 & 704,791,454.4 & \cite{EIA:2023:01}, \cite{EIA:2023:00}\\
\hline
\makecell{Generate Electric Power\\from All Sources: TN} & \makecell{Electric Power\\Demand: TNISO} & 574,599,600 & 1,127,281,017.6 & 552,681,417.6 & \cite{EIA:2023:01}, \cite{EIA:2023:00}\\
\hline
\makecell{Generate Electric Power\\from Summed Sources: TN} & \makecell{Electric Power\\Demand: TNISO} & 574,185,600 & 1,127,281,017.6 & 552,681,417.6 & \cite{EIA:2023:01}, \cite{EIA:2023:00}\\
\hline
\makecell{Generate Electric Power\\from All Sources: SOCO} & \makecell{Electric Power\\Demand: SOCO} & 1,888,477,200 & 1,685,908,461.6 & 202,568,738.4 & \cite{EIA:2023:01}, \cite{EIA:2023:00}\\
\hline
\makecell{Generate Electric Power\\from Summed Sources: SOCO} & \makecell{Electric Power\\Demand: SOCO} & 1,887,008,400 & 1,685,908,461.6 & 201,099,938.4 & \cite{EIA:2023:01}, \cite{EIA:2023:00}\\
\hline
\makecell{Generate Electric Power\\from All Sources: CARISO} & \makecell{Electric Power\\Demand: CARISO} & 1,624,262,400 & 1,554,163,243.2 & 70,099,156.8 & \cite{EIA:2023:01}, \cite{EIA:2023:00}\\
\hline
\makecell{Generate Electric Power\\from Summed Sources: CARISO} & \makecell{Electric Power\\Demand: CARISO} & 1,623,553,200 & 1,554,163,243.2 & 69,389,956.8 & \cite{EIA:2023:01}, \cite{EIA:2023:00}\\
\hline
\makecell{Generate Electric Power\\from All Sources: FPISO} & \makecell{Electric Power\\Demand: FPISO} & 1,795,406,400 & 1,764,674,834.4 & 30,731,565.6 & \cite{EIA:2023:01}, \cite{EIA:2023:00}\\
\hline
\makecell{Generate Electric Power\\from Summed Sources: FPISO} & \makecell{Electric Power\\Demand: FPISO} & 1,786,885,200 & 1,764,674,834.4 & 22,210,365.6 & \cite{EIA:2023:01}, \cite{EIA:2023:00}\\
\hline
\end{longtable}

\normalsize

\section{Conclusion}\label{ConclusionsCh7}

This paper uses a data-driven, MBSE-guided approach to develop the first open-source, asset-level granularity, behavioral model of the entire American Multi-Modal Energy System.  More specifically, socio-economic data on the flows of mass and energy through the AMES from the EIA are applied to the AMES reference architecture\cite{Thompson:2023:00} and instantiated architecture \cite{Thompson:2022:00} to create a behavioral model of the AMES.  The instantiated structural model includes the electric grid, the natural gas system, the oil system, the coal system, and the interconnections between them as defined by the AMES reference architecture, for the full contiguous United States of America (USA).  This behavioral model of the AMES recovers the relative weights of process flows and uses animated Sankey diagrams to visualize their evolution from month to month.  It also recovers the relative intensities of energy supply and demand and uses animated choropleth maps to visualize their evolution from month to month.   It also recovers the interstate network flows and uses interstate network maps to visualize their evolution from month to month.  In all three cases, these animated visualizations present aggregated results that draw upon the asset-level granularity of the AMES model.  Finally, this AMES model quantifies the weighted least square error associated with the input data.  Importantly, this AMES behavioral model is able to reconcile redundant and conflicting data on the EIA website.  

The AMES behavioral model relies on a weighted least square error hetero-functional graph state estimator (WLSEHFGSE) as its methodological underpinning.   The WLSEHFGSE developed here is the first to integrate state estimation methods with hetero-functional graph theory.  With respect to HFGT, the WLSEHFGSE is presented as a formal generalization of the hetero-functional network minimum cost flow (HFNMCF) problem (introduced in Sec. \ref{Sec:HFNMCF}).  Meanwhile, state estimation has played a major role in the operation and development of the American Electric Power System.  This work also proves that the WLSEHFGSE is a generalization of the weighted-least square error electric power system state estimator (WLSEEPSSE) problem (introduced in Sec.\ref{Sec:WLSEPSSE}).  As such, it extends electric power system state estimation to the heterogeneous multi-commodity conversions and flows in systems-of-systems like the AMES.   Furthermore, the WLSEHFGSE is fully extensible to the number of input datasets, number of systems in the systems-of-systems, the number assets, and the number of their associated capabilities.  Ultimately, the WLSEHFGSE allows direct comparisons in future studies between structural \cite{Thompson:2022:00,Thompson:2022:ISC-C78} and behavioral analyses; especially as these energy systems evolve.

Just as electric power system state estimation has played a major role in electric power system operation and markets, the AMES' state estimator can be expected to play an even greater role.  The evolution of the sustainable energy transition necessitates a profound understanding of critical energy infrastructures; their structure, behavior, and their coordination.  Such coordination prevents  cascading failures across the system-of-systems and facilitates optimal multi-commodity operations.  Through case studies in the Western Interconnect and the the Eastern Region, this work demonstrates how state estimation can be extended from the American electric power system to the AMES and still retain asset-level granularity.  The AMES' state estimator opens a plethora of future analyses.  These include resolving the EIA's target energy mixture and demands\cite{EIA:2020:00, EIA:2017:00, EIA:2016:00, EIA:2015:00, EIA:2014:12} to an asset level, responding to N-1 contingency scenarios, or analyzing new mult-energy system architectures that emphasize distribution, resilience, and technological innovation.  Ultimately, the WLSEHFGSE of the AMES has the potential to directly support and even enhance the EIA's Annual Energy Outlooks\cite{EIA:2020:00, EIA:2017:00, EIA:2016:00, EIA:2015:00, EIA:2014:12}.

\section*{Acknowledgments}\label{Sec:Acknowledgments}
This research is based on work supported by the National Science Foundation under Grant Numbers 2310638. 

\bibliographystyle{IEEEtran}
\bibliography{AMES_ML_Lib.bib}

\end{document}